\documentclass[aps, showpacs,preprintnumbers, superscriptaddress, twocolumn]{revtex4-2}
\usepackage{makeidx}
\usepackage{amsfonts}
\usepackage{amsmath}
\usepackage{amssymb}
\usepackage{eurosym}
\usepackage{graphicx}
\usepackage{array}
\usepackage{tensor}
\usepackage{longtable}
\usepackage{amsthm}
\usepackage{xcolor}
\usepackage{footmisc}
\newtheorem{definition}{Definition}[section]
\newtheorem{remark}{Remark}[section]
\newtheorem{theorem}{Theorem}[section]
\newtheorem{corollary}{Corollary}[section]
\newtheorem{proposition}{Proposition}[section]

\def\be{\begin{equation}}
\def\ee{\end{equation}}
\def\bea{\begin{eqnarray}}
\def\eea{\end{eqnarray}}

\begin{document}

\title{Length-preserving biconnection gravity and its cosmological implications}
\author{Lehel Csillag}
\email{lehel.csillag@ubbcluj.ro, lehel.csillag@unitbv.ro, lehel@csillag.ro}
\affiliation{Department of Physics, Babeș-Bolyai University, Kogălniceanu Street,
	Cluj-Napoca 400084, Romania,}
\affiliation{Department of Mathematics and Computer science, Transylvania University, Iuliu Maniu Street 50, Brașov 500091, Romania}
\author{Rattanasak Hama}
\email{rattanasak.h@psu.ac.th}
\affiliation{Faculty of Science and Industrial Technology, Prince of Songkla University,
Surat Thani Campus, Surat Thani, 84000, Thailand,}

\author{M\'{a}t\'{e} J\'{o}zsa}
 \email{mate.jozsa@ubbcluj.ro}
 \affiliation{Department of Physics, Babeș-Bolyai University, Kogălniceanu Street,
	Cluj-Napoca 400084, Romania,}
\author{Tiberiu Harko}
\email{tiberiu.harko@aira.astro.ro}
\affiliation{Department of Physics, Babeș-Bolyai University, Kogălniceanu Street,
	Cluj-Napoca 400084, Romania,}
\affiliation{Astronomical Observatory, 19 Cireșilor Street,
	Cluj-Napoca 400487, Romania,}
 
\author{Sorin V. Sabau}
\email{sorin@tokai.ac.jp}
\affiliation{School of Biological Sciences, Department of Biology and Graduate School of Science and Technology, Physical and Mathematical Sciences, Tokai University, Sapporo, 005-8600, Japan}

\date{\today }

\begin{abstract}
{We consider a length preserving biconnection gravitational theory, inspired by information geometry, which extends general relativity, by using the mutual curvature as the fundamental object describing gravity. The two connections used to build up the theory are the Schr\"{o}dinger connection, and its dual. In our geometric approach it can be seen that the dual of a non-metric Schr\"odinger connection  possesses torsion, even if the Schr\"odinger connection itself does not, and consequently the pair $(M,g,\nabla^{*})$ is a quasi-statistical manifold. The field equations are postulated to have the form of the standard Einstein equations},  but with the Ricci tensor- and scalar replaced with the mutual curvature tensor, and the mutual curvature scalar, resulting in additional torsion-dependent terms. The covariant divergence of the matter energy-momentum does not vanish in this theory. We derive the equation of motion for massive particles, which shows the presence of an extra force, depending on the torsion vector. The Newtonian limit of the equations of motion is also considered. We explore the cosmological implications by deriving the generalized Friedmann equations for the Friedmann-Lemaitre-Robertson-Walker (FLRW geometry). They contain additional terms that can be interpreted as describing an effective, geometric type dark energy. We examine two cosmological models: one with conserved matter, and one where dark energy and pressure are related by a linear equation of state. The predictions of both models are compared with  a set of observational values of the Hubble function, and with the standard $\Lambda$CDM model. Length-preserving biconnection gravity models fit well the observational data, and also align with $\Lambda$CDM at low redshifts $(z<3)$. {The obtained results suggest that a modified biconnection geometry could explain the late-time acceleration through an effective geometric dark energy, as well as the formation of the supermassive black holes, as they predict a different age of our Universe as compared to standard cosmology}.
\end{abstract}

\maketitle

\tableofcontents

\newpage
\section{Introduction}

The advent of the theory of general relativity \cite{Ein,Hilb} opened not only new perspectives in the understanding of gravity, one of the fundamental forces of nature, but also led to creative interplay and interaction between mathematics and physics. General relativity, based on a Riemannian geometric mathematical structure influenced the development of mathematics, leading to the creation of new research fields that also found many physical applications. Three years after general relativity was formulated as a consistent theory, Weyl \cite{Weyl1, Weyl2} proposed the first generalization of the Riemannian  geometry, based on the introduction of a new geometric concept, the nonmetricity of the spacetime.

For a long time Weyl's geometry did not find many applications in physics \cite{Scholz}, but it is presently an active field of research \cite{W1,W2}.  In the same year Weyl proposed his unified field theory based on the nonmetric geometry, Finsler \cite{Fins} introduced another extension of the Riemann geometry, in which the metric tensor is a function of both the coordinates and of a tangent vector. Finsler geometry has also the potential of opening some new directions in the understanding of quantum mechanics \cite{Fins1}, and of the gravitational phenomena \cite{Fins2}. Another important geometric concept is the torsion tensor of the spacetime, introduced by \'{E}lie Cartan \cite{C1,C2,C3,C4}, and which is at the basis of the Einstein-Cartan theory of the gravitational interaction \cite{C5}.

Among the interesting mathematical developments that greatly influenced physics one must also mention the work by Weitzenb\"{o}ck \cite{We1},  who introduced a geometry in which torsion exactly compensates curvature, and thus the spacetime is becoming flat. Weitzenb\"{o}ck's geometry represents the mathematical foundation of  $f(T)$ gravity \cite{We2}, and of its generalizations \cite{We3}. On the other hand, nonmetricity is the fundamental geometric quantity on which the $f(Q)$ gravity theories \cite{fQ1,fQ2,fQ3} are built up.

General relativity stands out in its description of the gravitational interaction compared to other fundamental interactions in particle physics. Unlike theories where forces result from the exchange of particles, like in electrodynamics, general relativity describes gravity as a property of spacetime geometry. Einstein's differential geometric description of gravitational interaction not only helped  to better understand  Riemannian geometry, but it also opened new possibilities for its generalization.

General relativity has been extensively tested at various length scales, and it has significant applications in astrophysics and cosmology. It gives an excellent description of the of the gravitational dynamics at the scale of the Solar System, and
explains well the perihelion precession of the planet Mercury, the bending of light by the Sun, and the Shapiro time delay effect \cite{O1,O2}. The detection of the gravitational waves \cite{O3,O4}  has  confirmed again the predictions of general relativity, and it has provided a new perspective for the analysis and description of the black hole - black hole, or black hole - neutron star merging processes.

The observational study  of the fluctuations in the temperature distribution of the Cosmic Microwave Background Radiation recently performed by the Planck satellite \cite{Pl1,Pl2}, as well as the investigations of the light curves of the Type Ia supernovae \cite{S1}, have convincingly shown that the present day Universe is in a state of accelerating cosmological expansion. Precise cosmological observations have also shown only around 5\% of the total matter-energy composition of the Universe amounts to  baryonic matter, with 95\% consisting  of two other, essentially unknown, constituents, commonly called dark energy, and dark matter, respectively. 

To describe phenomenologically the observational data obtained from the cosmological observations, the $\Lambda$CDM ($\Lambda$ Cold Dark Matter) model was proposed. This model is obtained by adding to the standard gravitational field equations $G_{\mu \nu}=\left(8\pi G/c^4\right)T_{\mu \nu}$, where $G_{\mu \nu}$ is the Einstein tensor, and $T_{\mu \nu}$ is the matter energy-momentum tensor, the cosmological constant $\Lambda$, first introduced by Einstein in 1917 in the gravitational field equations \cite{EinL} to obtain a static cosmological model of the Universe. However, after the discovery of the expansion of the Universe, Einstein dismissed the possibility of the presence of $\Lambda$.  Despite this, the $\Lambda$CDM model gives a very good description of the cosmological observational data at low redshifts, and thus it is considered as the standard cosmological paradigm of the present times.

The first problem the $\Lambda$CDM must face is related to the unknown nature, and physical/geometrical interpretation of the cosmological constant, which represents the so-called cosmological constant problem \cite{CC1,CC2}, whose solution is not yet known. An alternative view, which does not require the presence of $\Lambda$, is represented by the assumption that the Universe is filled with two components (of unknown physical origin), called dark energy and dark matter, respectively (for reviews of the dark anergy and dark matter problems see \cite{R1,R2,R3,R4,R5}). The $\Lambda$CDM model is a particular dark energy model, in which the effective pressure $p_{DE}$ and the effective energy $\rho_{DE}$ of the dark energy satisfies the equation of state $\rho_{DE}+p_{DE}=0$, $\forall t$.

The $\Lambda$CDM model faces several other important challenges. One of them, called the Hubble tension, is related to the existence of significant differences between the expansion rate of the Universe (the Hubble function) as obtained from the Cosmic microwave Background Radiation satellite observations, and from the low redshift determinations \cite{Val}.  As measured by the Planck satellite, the Hubble constant $H(0)$ has the value of $66.93 \pm 0.62$ km/ s/ Mpc \cite{Val}, while from the SHOES collaboration analysis its value is
$73.24 \pm 1.74$ km/ s/ Mpc \cite{Val}. Between these two values there is a difference which is more than $3\sigma$ \cite{Val}. If real, the Hubble tension strongly points  towards the need of extending, or even replacing the $\Lambda$CDM model.

A more significant problem is the James Webb Space Telescope (JWST) discovery of well-formed galaxies and supermassive black holes only a few hundred million years after the Big Bang \cite{Hain}. Unusually bright galaxy candidates have been detected at $z\approx 16$ \cite{Hari}. Additionally, polycyclic aromatic hydrocarbons (PAHs) have been identified at $z = 6.71$ \cite{Mun}.  The observed bright UV-irradiation of the early Universe suggests a reionization history that is much too short to satisfy
evolution of the hydrogen ionization fraction, $\chi_{HI}(z)$.  All these observational results  seriously challenge the timeline predicted by $\Lambda$CDM model.

The analysis of the measurements of the baryon acoustic oscillations (BAO) by the Dark Energy Spectroscopic Instrument (DESI) points towards a time-evolving equation of state of the dark energy \cite{Desi}. The BAO results have been obtained by using quasar and Lyman-$\alpha$ forest tracers in seven redshift bins in the redshift range $0.1 < z < 4.2$. While the DESI BAO data can be explained by the $\Lambda$CDM model, combining them with other observational datasets leads to results that contradict the standard cosmological scenario.

An attractive possibility of explaining the above problems of cosmology is to assume that general relativity, which can give an excellent description of the gravitational physics at the level of the Solar System is no longer valid on galactic or cosmological scales, and, in order to understand physics on large and very large scales a new theory of gravity is necessary. Many modified theories of gravity have been proposed (for reviews see \cite{MG1,MG2,MG3, MG4,MG5,MG6}), and they try to explain the gravitational phenomenology from different perspectives, by using new geometrical or physical structures that could explain the observational data. Extensions and maximal extensions of the Hilbert-Einstein action $S=\int{\left(R/2\kappa ^2+L_m\right)\sqrt{-g}d^4x}$ were considered in the framework of the $f(R)$ \cite{Bu}, $f\left(R,L_m\right)$ \cite{fRLm}, $f(R,T)$ theories \cite{fRT}, or of the Hybrid Metric Palatini Gravity theory \cite{Hyb}. For a review of the modified gravity theories with geometry-matter coupling see \cite{book}. Finsler type geometrical extensions of general relativity were considered in \cite{Fi1,Bouali_2023,Fins2}, while the astrophysical and cosmological effects of the Weyl geometry were investigated in \cite{Wey1,Wey2,Wey3,We4}.

An interesting approach to Weyl geometry was proposed by Schr\"{o}dinger \cite{Schrod}, who introduced a new type of connection. Although it contains nonmetricity, the Schr\"{o}dinger connection preserves the length of vectors under autoparallel transport. The physical and cosmological implications of the Schr\"{o}dinger connection were investigated in \cite{Klemm, Ming2024,csillag2024schrodinger}. {In \cite{Ming2024} the physical and cosmological implications of a Schr\"{o}dinger type connection of the form 
$\Gamma^\lambda_{~\mu\nu}=\overset{\circ}{\Gamma}\tensor{}{^\lambda_{\mu\nu}}-g^{\rho\lambda}Q_{\rho\mu\nu}$, where $\overset{\circ}{\Gamma} \tensor{}{^\lambda_{\mu\nu}}$ is the Levi Civita connection, while $Q_{\rho \mu \nu}$ is the nonmetricity, were investigated by using the variational principle for
\begin{align}\label{actionS}
S=\frac{1}{16\pi}\int d^4x& \sqrt{-g}\bigg(R+\frac{5}{24}Q_\rho Q^\rho+\frac{1}{6}\tilde{Q}_\rho\tilde{Q}^\rho+2T_\rho Q^\rho\nonumber\\
&+\zeta^{\rho\sigma}_{~~\alpha}T^\alpha_{~\rho\sigma}\bigg)+\int d^4x\sqrt{-g}L_m,
\end{align}
where $T^\mu_{~\rho\nu}$ is the torsion tensor, $T_\rho = T^\sigma_{~\rho\sigma}$, $\zeta^{\rho\sigma}_{~~\alpha}$ is a Lagrange multiplier, and $Q_\mu$ and $\tilde{Q}_\mu=-(1/2)Q_\mu$ are the nonmetricity vectors. The variation of the action with respect to the metric gives the field equations of the Weyl-Schr\"{o}dinger gravity as \cite{Ming2024}
\bea\label{eqWS}
\hspace{-0.5cm}\overset{\circ}{R}_{\mu\nu}-\frac{1}{2}\overset{\circ}{R}g_{\mu\nu}&-&\frac{2}{9}Q_\rho Q^\rho g_{\mu\nu}-\frac{11}{18}Q_\mu Q_\nu+\frac{2}{3}g_{\mu\nu}\overset{\circ}{\nabla}_\rho Q^\rho \nonumber\\
\hspace{-0.5cm}&+&\frac{1}{6}\left(\overset{\circ}{\nabla}_\nu Q_\mu+\overset{\circ}{\nabla}_\mu Q_\nu\right)=8\pi T_{\mu\nu},
\eea 
where $\overset{\circ}{R}_{\mu\nu}$ and $\overset{\circ}{R}$ are the Riemannian Ricci tensor and scalar, respectively, and $\overset{\circ}{\nabla}$ is the covariant derivative defined with respect to the Levi-Civita connection.
}

 The Friedmann-Schouten geometry and connection, in which the torsion has a specific form, and can be expressed in terms of a torsion vector, was investigated, from the point of view of the cosmological applications, in \cite{csillag2024semisymmetric}. Interestingly, symmetrizing a semi-symmetric type of torsion in a certain way, could lead to a Schrödinger connection, as pointed out in \cite{Klemm, csillag2024schrodinger}.

One of the open problems in present day theoretical physics is the problem of quantum gravity. Even a general theory of the quantized gravitational field, or of quantum geometry, does not yet exist, some relevant insights on the expected structure of the theory can be obtained from general physical considerations. Such a general prediction of quantum gravity models is the existence of a minimum length scale, which is assumed to be of the order of the Planck length $\sqrt{\hbar G/c^3}$ \cite{1,2, 2a}. In some quantum gravity models a minimum momentum scale is also considered \cite{3}, and a noncommutative approach to geometry may also be necessary \cite{NC}. 

An interesting approach to the problem of quantum geometries was proposed in \cite{4,4a}. In this model each point $\vec{r}$ in the classical background is associated with a vector $\left|g_{\vec{r}}\right>$ in a Hilbert space, $\vec{r}\rightarrow \left|g_{\vec{r}}\right>$, with
\begin{eqnarray} \label{|g_r>}
\left|g_{\vec{r}}\right> := \int g(\vec{r}{\, '}-\vec{r}) \left|\vec{r}{\, '}\right>  d^3\vec{r}{\, '} \, ,
\end{eqnarray}
where $g(\vec{r}{\, '}-\vec{r})$ is any normalised function. Hence, in this formulation of geometry of quantum mechanics, geometric (or physical) ``points'' in the background space do exist in a superposition of states, and as the result of measurements they may undergo {\it stochastic fluctuations}. Hence, the geometry of quantum mechanics can be formulated by associating to each spacetime point a {\it statistical distribution}.

From a mathematical point of view such mathematical structures are well known, and they are called {\it statistical manifolds} (for a detailed discussion of the subject see the books by Amari \cite{5,5a}). Statistical manifolds represent an application of Riemannian geometry to the study of stochastic processes. The concept was initially introduced by Lauritzen \cite{6}, and it was later reformulated by Kurose in \cite{7} from the perspective of the affine differential geometry. For short, a statistical manifold $(M, \nabla,h)$ is a (semi-)Riemannian manifold $(M, h)$ with a torsion-free affine connection $\nabla$ with $\nabla h$  totally symmetric \cite{8}. For a statistical manifold  a pair of mutually dual affine connections can be naturally defined. Thus, if a statistical manifold $(M, \nabla, h)$ is given, and we denote by $\nabla ^{*}$ the dual connection of $\nabla$ with respect to $h$, the triplet $( M, \nabla ^{*}, h)$ is also a statistical manifold. $( M, \nabla ^{*}, h)$ is called the dual statistical manifold of $(M, \nabla,  h)$ \cite{8}.

Statistical manifolds and dual affine connections were rediscovered in statistics in order to build up geometric theories for statistical inferences \cite{5,5a}. Presently,  this geometric method is called information geometry, and is applied in various fields of mathematical sciences \cite{9,10}.

{Kurose \cite{kurose2007} and Matsuoze \cite{8} pointed out that for the description of certain quantum effects, the inclusion of torsion might be needed due to the non-commutativity of quantum mechanics. With this motivation at hand, Kurose introduced the notion of a statistical manifold admitting torsion, or a quasi-statistical manifold. However, there is a sharp contrast in this case, compared to the case of simple statistical manifolds. Given a quasi-statistical manifold $(M,g,\nabla)$, if we replace the $\nabla$ with its dual, generally, it will not be true that $(M,g,\nabla^{*})$ is a quasi-statistical manifold as well. The exact mathematical conditions, which provide the conditions on $T$ and $T^{*}$ for the dual to be a quasi-statistical manifold are given in corollary \ref{corollarytorsionalmanifolds} { (which is also present in different notations in the mathematical paper \cite{Blaga2022}}). {In some more detail, if the connection $\nabla$ is torsion-free, then the pair $(M,g,\nabla^{*})$ formed by the dual of the connection is a quasi-statistical manifold, but $T^{*}$ need not necessarily be torsion-free, and conversely, if $\nabla^{*}$ is torsion-free, then the pair $(M,g,\nabla)$ is a quasi-statistical manifold, but $T$ does not necessarily have to be torsion-free.}}

Although the link to statistical manifolds and probabilistic geometry was only recently discovered (as we will explain later), gravitational theories based on two connections have already been investigated \cite{Khosravi_multi, Khosravi_2014, Khosravi_2015, Tamanini:2013mg, Tamanini:2012hg}. In \cite{Khosravi_multi} the geodesic
equation was generalized to
\begin{eqnarray}
\frac{d^2x^\mu}{d\lambda^2}&=&-\sum_{i=1}^N
{^{(i)}}\Gamma^{\mu}_{\alpha\beta}\frac{dx^\alpha}{d\lambda}\frac{dx^\beta}{d\lambda}
=
-N {\overline{\Gamma}^{\mu}_{\alpha\beta}}\frac{dx^\alpha}{d\lambda}\frac{dx^\beta}{d\lambda},
\end{eqnarray}
where  $(i)$ labels the number of connections, and the
{weighted} average connection is defined as ${\overline{\Gamma}^{\mu}_{\alpha\beta}}\equiv
\frac{1}{N}\sum_{i=1}^N \limits {^{(i)}}\Gamma^{\mu}_{\alpha\beta}$. For gravitational applications the Hilbert-Einstein action is generalized as
$
{\cal{S}}=\int d^4x \sqrt{-g}g^{\mu\nu}\frac{1}{N}\sum_{i=1}^N
R_{\mu\nu}\left(^{(i)}\Gamma^\rho_{\alpha\beta}\right),
$
which in the case of a biconnection model reduces to
$
{\cal{S}}=\int d^4x {\cal{L}}=\int d^4x
\sqrt{g}g^{\mu\nu}\frac{1}{2}\left[
R_{\mu\nu}\left(^{(1)}\Gamma^\rho_{\alpha\beta}\right)+
R_{\mu\nu}\left(^{(2)}\Gamma^\rho_{\alpha\beta}\right)\right]$ \cite{Khosravi_multi},
which can be reformulated as
\begin{eqnarray}\label{indices3}
{\cal{L}}=\sqrt{g} g^{\mu\nu}
\left[R_{\mu\nu}\left({\overline{\Gamma}^\rho_{\alpha\beta}}\right)+\Omega^{\alpha}_{\alpha\lambda}
\Omega^{\lambda}_{\nu\mu}-\Omega^{\alpha}_{\nu\lambda}\Omega^{\lambda}_{\alpha\mu}\right],
\end{eqnarray}
where $R_{\mu\nu}(\gamma^\rho_{\alpha\beta})$ denotes the Ricci tensor constructed from
 the average connection ${\overline{\Gamma}^\rho_{\alpha\beta}}\equiv\frac{1}{2}\left(^{(1)}\Gamma^\rho_{\alpha\beta}+
^{(2)}\Gamma^\rho_{\alpha\beta}\right)$,   and
$\Omega^\rho_{\alpha\beta}\equiv\frac{1}{2}\left(^{(1)}\Gamma^\rho_{\alpha\beta}-
^{(2)}\Gamma^\rho_{\alpha\beta}\right)$, is a tensor obtained from the transformation rule of the connections. {It is important to mention that in this approach, both the connections were assumed to be symmetric, hence the $\Omega$ tensors are symmetric in their lower indices.} The approach initiated in \cite{Khosravi_multi} was generalized in \cite{Khosravi_2014} by assuming that the two connections are defined in a Weyl geometry, and thus they satisfy the relations $^{(1)}\nabla _\mu g_{\alpha \beta}=-C_\mu g_{\alpha \beta}$, and $^{(2)}\nabla _\mu g_{\alpha \beta}=+C_\mu g_{\alpha \beta}$, respectively. The Weyl biconnection model is a natural framework to generate the mathematical structure of a Galileon theory. This model also admits a self-accelerating solution, and is closely related to massive gravity in the multiconnection framework.

{{In \cite{Iosifidis2023} special type of biconnection a biconnection theory was studied, described by the action $S=(1/4\kappa)\int{\left(R^{(1)}+R^{(2)}+K\right)\sqrt{-g}d^4x}$, with $K:=K^{\lambda \mu \nu}K_{\mu \nu \lambda}-K^{\lambda \mu}_{\;\;\;\;\mu}K^{\lambda \nu}_{\;\;\;\nu}$ being the difference scalar. In this framework, the connections are sourced by hypermomentum, exhibit zero torsion, and have a completely symmetric nonmetricity. However, when considering the vacuum case, where the coupling between connection and matter (via hypermomentum) is not included, this theory is indistinguishable from GR.}
{
Once the matter-connection coupling is taken into account, the theory introduces two hypermomenta, defined as: $\Delta _\lambda ^{\;\;\;\mu \nu (1)}=-2/\sqrt{-g} \delta S_m/\delta \Gamma ^{\lambda\;\;\;\;\;(1)}_{\;\;\;\mu \nu}=\Xi _\lambda ^{\;\;\mu \nu}$, and $\Delta _\lambda ^{\;\;\;\mu \nu (2)}=-2/\sqrt{-g} \delta S_m/\delta \Gamma ^{\lambda\;\;\;\;\;(2)}_{\;\;\;\mu \nu}=-\Xi _\lambda ^{\;\;\mu \nu}$, respectively. Assuming the hypermomentum tensor $\Xi_{\alpha \mu \nu}$ is totally symmetric and traceless, it follows that the connection coefficients can be written as $\Gamma ^{\lambda\;\;\;\;(1)} _{\;\;\; \mu \nu}=\overset{\circ}{\Gamma} \tensor{}{^\lambda _\mu _\nu} +\kappa \Xi ^{\lambda}_{\;\;\mu \nu}$, and $\Gamma ^{\lambda\;\;\;\;(2)} _{\;\;\; \mu \nu}=\overset{\circ}{\Gamma} \tensor{}{^\lambda _\mu _\nu} -\kappa \Xi ^{\lambda}_{\;\;\mu \nu}$, respectively. The field equations for this model take the form \cite{Iosifidis2023}
\be
\overset{\circ}{R}_{\mu \nu}-\frac{1}{2}g_{\mu \nu}\overset{\circ}{R}=\kappa T_{\mu \nu}-\kappa ^2\left(\Xi ^{\alpha \beta}_{\;\;\mu}\Xi _{\alpha \beta \mu}-\frac{1}{2}\Xi^{\alpha \beta \gamma}\Xi_{\alpha \beta \gamma}g_{\mu \nu}\right).
\ee
}}

{Hence, the gravitational theory based on two connections in the metric-affine framework possesses the structure of a statistical manifold, thanks to the symmetries of the associated hypermomentum. {Even though the connection-matter coupling via hypermomentum is not yet well justified in nature, it has some theoretical motivation, such as formulating a gauge theory of gravity \cite{Hehl:1994ue}. For more details on hypermomentum, consult \cite{hehl1976hypermomentum}. }}

{
A key and novel ingredient in achieving this result is the mutual curvature scalar, a newly defined object from the mutual curvature tensor. This mutual curvature scalar, contains the difference scalar $K$, which introduces a non-trivial coupling between the two connections. It is important to note, however, that the mutual curvature tensor has a relatively rich history, at least from a mathematical perspective.  In \cite{iosifidis2023torsioncurvatureanaloguedualconnections} it was shown that the mutual curvature tensor, as traditionally used by mathematicians \cite{puechmorel2020lifting, calin2014geometric}, does not actually meet the criteria for being a true tensor, due to its lack of multilinearity in each of its slots. }{More precisely, mathematicians considered the mutual curvature of two connections $\nabla^{(1)}$ and $\nabla^{(2)}$ to be defined according to \cite{puechmorel2020lifting}
\begin{equation}
    R_{\nabla_1 \nabla_2}(X,Y)Z:=\nabla_X^{(1)} \nabla_{Y}^{(2)} Z - \nabla_{Y}^{(1)} \nabla_{X}^{(2)}Z - \nabla_{[X,Y]}^{(1)}Z.
\end{equation}
However, this expression is not $C^{\infty}(M)$-linear in its third slot ($Z$), so it cannot be a true tensor \cite{iosifidis2023torsioncurvatureanaloguedualconnections}. Iosifidis \cite{iosifidis2023torsioncurvatureanaloguedualconnections} refined this notion, by defining \begin{equation} \label{curvaturemutualcoofree}
\begin{aligned}
\mathcal{R}_{\nabla_1,\nabla_2}(X,Y)Z=\frac{1}{2} \Big(&\nabla_{X}^{(1)} \nabla_{Y}^{(2)} Z - \nabla_Y^{(1)} \nabla_{X}^{(2)}Z \\
&+ \nabla_{X}^{(2)} \nabla_Y^{(1)}Z - \nabla_Y^{(2)} \nabla_X^{(1)}Z\\
&-\nabla^{(1)}_{[X,Y]} Z - \nabla^{(2)}_{[X,Y]}Z\Big).
\end{aligned}
\end{equation}}
{He has also shown that this object is a true tensor, as it is $C^{\infty}(M)$ multilinear in each of its slots. It also treats both of the connections on an equal footing, since it is symmetric under the exchange $\nabla^{(1)} \leftrightarrow \nabla^{(2)}$. In coordinates, the components of this tensor read
\begin{equation}
\begin{aligned}
\mathcal{R}\tensor{}{^\lambda _\rho _\mu _\nu}=&\frac{1}{2} \left(\tensor{{R}}{^\lambda _\rho _\mu _\nu}^{(1)}+\tensor{{R}}{^\lambda _\rho _\mu _\nu}^{(2)} \right)\\
&- \frac{1}{2} \tensor{K}{^\lambda _\sigma _\mu} \tensor{K}{^\sigma_\rho_\nu} + \frac{1}{2} \tensor{K}{^\lambda _\sigma _\nu} \tensor{K}{^\sigma _\rho _\mu},
\end{aligned}
\end{equation}
where $K$ is given by
\begin{equation}
\tensor{K}{^\rho _\alpha _\beta}=\tensor{\Gamma}{^\rho _\alpha _\beta}^{(1)} - \tensor{\Gamma}{^\rho _\alpha _\beta}^{(2)}.
\end{equation}
By contracting $\lambda$ with $\mu$, one immediately obtains
\begin{equation}\label{indices2}
\mathcal{R}_{\rho \nu} = \frac{1}{2} \left( \tensor{R}{_{\rho \nu}}^{(1)} + \tensor{R}{_{\rho \nu}}^{(2)} \right) 
   - \frac{1}{2} \tensor{K}{^\lambda _\sigma _\lambda} \tensor{K}{^\sigma _\rho _\nu} 
   + \frac{1}{2} \tensor{K}{^\lambda _\sigma _\nu} \tensor{K}{^\sigma _\rho _\lambda}.
\end{equation}
Although it might simply seem that
\begin{equation}
    \tensor{\Omega}{^\rho _\alpha _\beta}=\frac{1}{2} \tensor{K}{^\rho _\alpha _\beta},
\end{equation}
this is not necessarily true, only in a special case. Note that the order of indices in expressions \eqref{indices2} and \eqref{indices3} differs. However, for the case in which $K$ is symmetric (as was considered in \cite{Khosravi_multi}), the two curvature tensors do agree. Hence, the approach of Iosifidis presented in \cite{iosifidis2023torsioncurvatureanaloguedualconnections} can be seen as a generalization of the previously proposed decomposition in \cite{Khosravi_multi} to a case where torsion is not necessarily vanishing. Even if this seems trivial from the formulas in local coordinates, the coordinate-free expression \eqref{curvaturemutualcoofree} was essential for mathematicians to take into account that the object they used before \cite{puechmorel2020lifting, calin2014geometric} is not a true tensor, while the one described by \eqref{curvaturemutualcoofree} is. This generalization has led to follow-up works both in physics \cite{Iosifidis2023} and mathematics \cite{axioms12070667}.
}

It is the goal of the present paper to explore the physical and the cosmological implications of a biconnection theory, based on the newly defined mutual curvature scalar, and by using a metric approach {in contrast to the metric-affine approach considered in \cite{Iosifidis2023}.} {\it We will specifically fix the two connections to be the Schr\"{o}dinger connection, and its dual, {without considering a geometry-matter  coupling.}} This choice is physically justified, since the Schr\"{o}dinger connection preserves the lengths of autoparallelly transported vectors, even though it is not metric-compatible. Considering the dual of this connection {is a novel and interesting method} to introduce a semi-symmetric type of torsion into the dual geometry. After introducing the basic geometric concepts used for the description of a statistical manifold with torsion, {we show that for a Schrödinger-type connection, the pair $(M,g,\nabla^{*})$ is a quasi-statistical manifold. We postulate that the field equations of the proposed biconnection theory take the same form as in the Einstein theory, with the only difference being that the Ricci tensor $R_{\mu \nu}$ and Ricci scalar $R$ are replaced by the mutual curvature $\mathcal{R}_{\mu \nu}$ and mutual curvature scalar $\mathcal{R}$, with the two connections considered to be the Schrödinger connection, and its dual.} We also obtain the equation of motion of massive test particles, which is generally non-geodesic, and takes place in the presence of an extra force, which is fully determined by the torsion vector. The Newtonian limit of the equation of motion is also considered.

We also perform a detailed analysis of the cosmological implications of the theory. As a first step in this study we derive the generalized Friedmann equations for a homogeneous, isotropic and flat geometry, which contains extra torsion/nonmetricity dependent terms, which we interpret as corresponding to the effective energy density and pressure of the dark energy. To close the system of cosmological equations we need to impose some extra conditions on the model parameters. We explore two such models: one imposes the condition of matter energy conservation, and the other assumes a linear equation of state relating the effective dark energy pressure and energy density. A detailed comparison of the models with a small set of observational data for the Hubble function, and with the $\Lambda$CDM paradigm is performed. This comparison indicates that the considered cosmological models may represent some viable alternatives to the $\Lambda$CDM paradigm.

The present paper is organized as follows. {After briefly introducing the geometric perspective of the well-known statistical manifolds, we present the less known notion of a quasi-statistical manifold. In particular, we show that for Weyl or Schrödinger connections $\nabla$, the dual connection $\nabla^{*}$ is not torsion-free, but nevertheless, the pairs $\left(M,g,\nabla^{*} \right)$ are statistical manifolds admitting torsion. In Section \ref{sect2} we propose a biconnection gravity model in the metric formalism, using the recently defined mutual curvature tensor of \cite{Iosifidis2023}, by postulating that the field equations take the form 
\begin{equation}
    \mathcal{R}_{(\mu \nu)}-\frac{1}{2} g_{\mu \nu} \mathcal{R}=8\pi T_{\mu \nu}.
\end{equation}
}

{
As we work in the metric formalism, the connections have to be specified as well: we choose the Schrödinger connection and its dual, thanks to their physically reasonable length-preserving properties. We also study some physical applications of the proposed length-preserving biconnection theory, which differs significantly from usual GR. A main difference is the non-conservation of the energy-momentum  tensor, which we attribute to the presence of an extra force. We also obtain the equation of motion of the massive particles, and the Newtonian  limit of the equations of motion.} The cosmological implications of the theory are investigated in Section~\ref{sect3}, where the generalized Friedmann equations are obtained, and two cosmological models are introduced. The predictions of the models are compared with the observations, and with the predictions of the $\Lambda$CDM model, in Section~\ref{sect4}. Finally, we discuss and conclude our results in Section~\ref{sect5}. {For the sake of completeness, we present a rigorous coordinate-free mathematical treatment of the geometry of quasi statistical manifolds Appendix in ~\ref{appendixA}}. The details of the calculation of the Ricci and mutual difference tensors are presented in Appendix~\ref{appendixB}. The derivation of the generalized Friedmann equations is shown in Appendix~\ref{appendixC}.

\section{(Dual) Statistical structures on manifolds equipped with two connections}\label{sect1}

{In this section, we briefly summarize the basic statistical structures, which can be put on a manifold $M$, given an affine connection $\nabla$. We first introduce the notion of a statistical manifold $(M,g,\nabla)$ and its dual connection $\nabla^{*}$, assuming the absence of torsion. Then, we generalize these constructions for the case of connections with torsion, leading to the notion of a \textit{quasi-statistical manifold}. We show that by considering a Schrödinger or Weyl connection, even though the pairs $(M,g,\nabla)$ are not statistical manifolds, it holds true that $(M,g,\nabla^{*})$ are quasi-statistical manifolds, or statistical manifolds admitting torsion, dubbed \textit{Quasi-Weyl} and \textit{Quasi-Schrödinger} manifolds, respectively. Thus, considering the dual connection $\nabla^{*}$ offers an interesting novel way to introduce torsion into these geometries.}

\subsection{Statistical manifolds}

{In the following, we will introduce the notion of a statistical manifold, assuming a torsion-free connection. This notion is based on a particular geometry, which generalizes the Riemannian one.}

{Let us consider an $n$-dimensional (pseudo)-Riemannian manifold $M$ with local coordinates $ \{ x^{\mu} \}$ and a torsion-free affine connection $\nabla$, described by the coefficients 
\begin{equation}
    g \left( \nabla_{\partial_\mu} \partial_\nu, \partial_\rho \right)=\Gamma_{\rho \nu \mu}.
\end{equation}
}

{
If there exists a totally symmetric tensor $C_{\mu \nu \rho}$ (often called the cubic tensor), such that:
\begin{equation}\label{statistical}
    \nabla_{\mu} g_{\nu \rho}=C_{\mu \nu \rho},
\end{equation}
 then the pair $(M,g,\nabla)$ is called a \textit{statistical manifold}. The historical origin of the name is tied to the close relation of this specific geometry with statistics. To illustrate this, let us consider a family of probability distributions $p=p(x,\sigma)$ with the normalization condition $\int p(x,\sigma) dx=1$. Supposing that $p(x,\sigma)$ depends smoothly on the parameters $\sigma=(\sigma_1,\sigma_2,\dots,\sigma_n)$, these probability distributions form a differentiable manifold. Moreover, we can also define a natural metric on them, known as the \textit{Fisher metric}, given by:
\begin{equation}
    g_{ij}(\sigma):=E_{\sigma}[\partial_i l, \partial_j l]=\int p(x,\sigma) \partial_i l(x,\sigma) \partial_j l(x,\sigma) dx,
\end{equation}
where $E_{\sigma}[f]=\int f(x,\sigma) p(x,\sigma) dx$ is the expectation value of a function $f$, and $l(x,\sigma):=\ln p(x,\sigma)$ denotes the log-likelihood function. Given this data, a natural cubic tensor 
\begin{equation}
    C_{ijk}(\sigma):=\int p(x,\sigma) \partial_i l(x,\sigma) \partial_j l(x,\sigma) \partial_k l(x,\sigma) dx
\end{equation}
can be defined, which makes the tuple $(M,g,\nabla)$ a statistical manifold. However, we would like to point out that in the following, we will simply refer to statistical manifolds from a completely geometric perspecive, not necessarily tied to statistics.
}

{
Mathematicians observed that if we have a statistical manifold $(M,g,\nabla)$, there exists another connection $\nabla^{*}$ such that the pair $(M,g,\nabla^{*})$ also forms a statistical manifold (often referred to as the dual statistical manifold). This special connection is known as the \textit{dual connection}. The connection coefficients of $\nabla^{*}$ are given by \cite{WSM1}
\begin{equation}
\partial_\mu g_{\nu \rho} -g_{\beta \rho}\tensor{\Gamma}{^\beta _\nu _\mu} - g_{\beta \nu} \tensor{\Gamma}{^{\beta*} _\rho _\mu}=0.  \label{1}
\end{equation}
}

The above equation implies that the inner product between two vectors $%
A_{\mu }$ and $B_{\mu }$ is preserved under the parallel transport following the
two dual affine connection
\begin{equation}
g_{\mu \nu }A^{\mu }(x)B^{\nu }(x)=g_{\mu \nu }\left( x+dx\right) A^{\mu
}\left( x+dx\right) B^{\nu \ast }\left( x+dx\right) ,
\end{equation}%
where $A^{\mu }\left( x+dx\right) $ and $B^{\nu \ast }\left( x+dx\right) $
are parallelly transported vectors by the dual affine connections $\tensor{\Gamma}{^\beta_\nu _\mu}$ and $\tensor{\Gamma}{^{\beta \ast}_{\rho \nu}}$, respectively, so that
\begin{equation}
\delta A^{\mu }=-\tensor{\Gamma}{^\mu _\nu _\rho}A^{\nu }dx^{\rho },\delta B^{\mu
}=-\tensor{\Gamma}{^{\mu \ast}_{\nu \rho} }B^{\nu }dx^{\rho }.
\end{equation}

{
{In modified gravity, it is a standard result that a general affine connection can be decomposed according to \cite{csillag2024semisymmetric}:
\begin{equation}\label{generalconnection}
\begin{aligned}
    \tensor{{\Gamma}}{^\mu _\nu _\rho}=\overset{\circ}{\Gamma}\tensor{}{^\mu _\nu _\rho} &+ \frac{1}{2} g^{\lambda \mu}(-Q_{\lambda \nu \rho}+ Q_{\rho \lambda \nu} + Q_{\nu \rho \lambda})\\
    &- \frac{1}{2}g^{\lambda \mu}(T_{\rho \nu \lambda}+T_{\nu \rho \lambda}- T_{\lambda \rho \nu}),
    \end{aligned}
\end{equation}
where $\overset{\circ}{\Gamma}$ denotes the connection coefficient functions of the Levi-Civita connection, and the nonmetricity and torsion are given by
\begin{equation}
    Q_{\mu \nu \rho}=-\nabla_{\mu} g_{\nu \rho}, \; \; \tensor{T}{^\mu _\nu _\rho}=2 \tensor{\Gamma}{^\mu _{[\rho \nu]}}.
\end{equation}}

{As we assumed the connection to be torsion-free and total symmetry of $C$ in \eqref{statistical}  (or equivalently of $Q$ by the definition of a statistical manifold), we immediately obtain
\begin{equation}
    \tensor{\Gamma}{^\mu _\nu _\rho}=\overset{\circ}{\Gamma} \tensor{}{^\mu _\nu _\rho}- \frac{1}{2} \tensor{C}{^\mu _\nu _\rho}.
\end{equation}}

{
According to \eqref{dualcoeffs}, since $Q_{\mu \nu \rho}$ is completely symmetric and $T_{\rho \mu \nu}=0$, we readily get
\begin{equation}
    Q^{*}_{\mu \nu \rho}=C_{\mu \nu \rho}, \; \; T^{*}_{\rho \mu \nu}=0,
\end{equation}
which easily yields
\begin{equation}
    \tensor{\Gamma}{^{\mu *} _\nu _\rho}=\overset{\circ}{\Gamma} \tensor{}{^\mu _\nu _\rho} + \frac{1}{2} \tensor{C}{^\mu _\nu _\rho}.
\end{equation}}

{
It is thus clearly seen that the average connection is the Levi-Civita connection
\begin{equation}
   \overline{\Gamma} \tensor{}{^\mu _\nu _\rho}=\frac{1}{2} \left(\tensor{\Gamma}{^\mu _\nu _\rho}+ \tensor{\Gamma}{^{\mu *}_\nu_\rho} \right)=\overset{\circ}{\Gamma}\tensor{}{^\mu _\nu _\rho},
\end{equation}
where $\overline{\Gamma}\tensor{}{^\mu_\nu_\rho}$ denotes connection coefficient functions of the average connection, and $\overset{\circ}{\Gamma} \tensor{}{^\mu _\nu _\rho}$ are the Christoffels symbols of the Levi-Civita connection. Note that the assumption on $\Gamma$ being torsion-free was not sufficient to conclude that $T^{*}=0$, we also used that the nonmetricity was completely symmetric, i.e. \eqref{statistical}.}

{Let us also present a direct coordinate-free argument of this statement.
\begin{proposition}\label{prop1}
Let $\nabla$ and $\nabla^{*}$ be dual connections on a manifold $(M,g)$. Then the average connection
\begin{equation}
    \overline{\nabla}=\frac{1}{2} \left(\nabla + \nabla^{*} \right)
\end{equation}
is $g$-metrical, i.e. $(\overline{\nabla}_X g)(Y,Z)=0$.
\end{proposition}
\begin{proposition}\label{prop2}
    If $(M,g,\nabla)$ is a statistical manifold, i.e. $\nabla$ is torsion-free and the cubic tensor $C(X,Y,Z)=(\nabla_X g)(Y,Z)$ is totally symmetric, then the dual connection $\nabla^{*}$ is torsion-free.
\end{proposition}
The self contained proofs of these two statements are given in Appendix \ref{appendixA}.  Combining them, we immediately obtain the desired result, as shown in the following corollary.
\begin{corollary}
    Let $(M,g,\nabla)$ be a statistical manifold. Then, the average connection
    \begin{equation}
        \overline{\nabla}=\frac{1}{2} \left(\nabla+\nabla^{*} \right)
    \end{equation}
    coincides with the Levi-Civita connection.
\end{corollary}
\begin{proof}
    From Proposition  \ref{prop1} it follows that $\overline{\nabla}$ is metric-compatible with respect to $g$. From Proposition \ref{prop2} it follows that $\nabla^{*}$ is torsion-free. Using
    \begin{equation}
        \overline{T}(X,Y)=\frac{1}{2}\left( T(X,Y)+T^{*}(X,Y) \right),
    \end{equation}
    we immediately get that $\overline{T}=0$, concluding the proof.
\end{proof}
}

\subsection{Quasi-statistical manifolds}

In his seminal paper \cite{kurose2007}, Takashi Kurose introduced the concept of a \textit{statistical manifold admitting torsion}, also known as a \textit{quasi-statistical manifold}. This was motivated by the idea of relating the non-commutativity of quantum mechanics to torsion.

Simply put, a quasi-statistical manifold is a manifold $M$ equipped with an affine connection $\nabla$ with torsion, satisfying the local condition:
\begin{equation}\label{quasistatistical}
    \nabla_\mu g_{\nu \rho}-\nabla_\nu g_{\mu \rho}+ T_{\rho \mu \nu}=0.
\end{equation}

With the help of the nonmetricity tensor, we can rewrite \eqref{quasistatistical} as
\begin{equation}\label{conditionquasi}
    - Q_{\mu \nu \rho} + Q_{\nu \mu \rho} + T_{\rho \mu \nu}=0.
\end{equation}

One is naturally led to ask the question: given the nonmetricity $Q$ and torsion $T$ of $\nabla$, is it possible to find the nonmetricity $Q^{*}$ and torsion $T^{*}$ of $\nabla^{*}$? Equation \eqref{dualcoeffs} provides the following positive answer:
\begin{equation}\label{torsiondual}
    Q^{*}_{\mu \nu \rho}=-Q_{\mu \nu \rho}, \; \; T^{*}_{\rho \mu \nu}= T_{\rho \mu \nu}- Q_{\mu \nu \rho} + Q_{\nu \mu \rho}.
\end{equation}

It is important to note that even if $\nabla$ is torsion-free, its dual $\nabla^{*}$ could have torsion. Hence, considering the dual of a connection could be seen as an information-geometric procedure to generate torsion. Moreover, from equation \eqref{torsiondual} we can observe that if $\nabla$ is torsion-free, then the equation
\begin{equation}
    T^{*}_{\rho \mu \nu} + Q_{\mu \nu \rho} - Q_{\nu \mu \rho}=0
\end{equation}
is satisfied. Relating the nonmetricity $Q$ to the nonmetricity $Q^{*}$ of the dual connection, we obtain
\begin{equation}
    T^{*}_{\rho \mu \nu}- Q^{*}_{\mu \nu \rho}+ Q^{*}_{\nu \mu \rho}=0.
\end{equation}

It immediately follows that the dual connection $\nabla^{*}$ satisfies the condition \eqref{conditionquasi}. Altogether, we conclude that it is a statistical manifold admitting torsion.
{
\begin{remark}
    As Corollary \ref{corollarytorsionalmanifolds} shows, this works the other way as well. The dual connection $\nabla^{*}$ is torsion-free if and only if the pair $(M,g,\nabla)$ is a statistical manifold admitting torsion.
\end{remark}
}
We present two quasi-statistical manifolds, which arise from vectorial nonmetricities. They have the special property that $Q_{\mu \nu \rho}$ is fully determined by a vector.

\textbf{Quasi-Weyl manifolds.} We consider a manifold $M$ equipped with a torsion-free affine connection with Weyl nonmetricity
\begin{equation}
    Q_{\mu \nu \rho}= W_{\mu} g_{\nu \rho}.
\end{equation}

For the dual connection we obtain
\begin{equation}
    Q^{*}_{\mu \nu \rho}=- W_{\mu} g_{\nu \rho},\; \;   T^{*}_{\rho \mu \nu}=- W_{\mu} g_{\nu \rho} + W_{\nu} g_{\mu \rho}.
\end{equation}

We observe that the dual connection has a semi-symmetric type of torsion. Hence, the procedure of dualizing a Weyl connection introduces a semi-symmetric type of torsion into this geometry.

\textbf{Quasi-Schr\"{o}dinger manifolds.} Let us consider a Schr\"{o}dinger connection \cite{Klemm, Ming2024, csillag2024schrodinger}, which is torsion-free and has nonmetricity of the form
\begin{equation}
    Q_{\mu \nu \rho}=\pi_{\mu} g_{\nu \rho}- \frac{1}{2} \left(g_{\mu \nu} \pi_\rho + g_{\mu \rho} \pi_\nu \right).
\end{equation}

In this case, the dual connection is specified by
\begin{equation}
\begin{aligned}
    Q^{*}_{\mu \nu \rho}&= -\pi_{\mu} g_{\nu \rho} + \frac{1}{2} \left( g_{\mu \nu} \pi_\rho + g_{\mu \rho} \pi_{\nu} \right),\\
    T^{*}_{\rho \mu \nu}&=\frac{3}{2} \left(\pi_{\nu} g_{\mu \rho} - \pi_{\mu} g_{\nu \rho} \right).
\end{aligned}
\end{equation}

Notably, not only the parallel transport along $\nabla$ and $\nabla^{*}$ combined preserves lengths in Schrödinger geometry. With the help of this connection, it is widely known that the lengths of vectors are preserved under autoparallel transport \cite{Schrod,Klemm, Ming2024, csillag2024schrodinger}. Since the dual connection $\nabla^{*}$ also has this type of nonmetricity up to a rescaling with $-1$, it follows that it preserves the lengths of autoparallely transported vectors as well.

A summary of our findings related to quasi-statistical manifolds with vectorial nonmetricity is provided in Fig.~\ref{fig1}.
\begin{figure*}[htbp]
\centering
\includegraphics[width=1\linewidth]{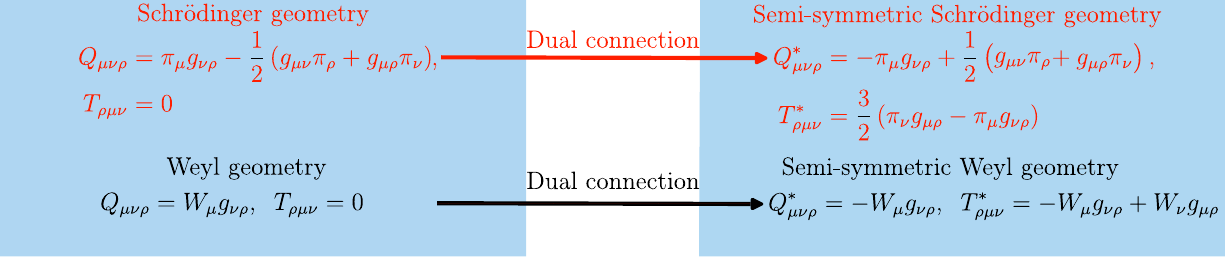}
\caption{Illustration of the dual connections associated to vectorial nonmetricities.}
\label{fig1}
\end{figure*}

\section{Length-preserving biconnection gravity}\label{sect2}

{In this section, we describe the geometric way to couple the Schrödinger connection to its dual, using the recently defined mutual curvature tensor \cite{Iosifidis2023}. Then, the generalized Einstein equations are presented, and some of their properties  are investigated.}

\subsection{Mutual curvature of Schrödinger and quasi-Schrödinger geometries}

In \cite{Iosifidis2023} it has been shown that a gravitational theory based on two connections can exhibit the structure of a statistical manifold, if considered in the metric-affine approach. The building pillar of that theory is the mutual curvature tensor, which is defined as
\begin{eqnarray}\label{mutualcurvaturetensor}
\tensor{\mathcal{R}}{^\lambda_\rho _\mu _\nu}&=&\frac{1}{2} \left( \tensor{R}{^\lambda _\rho _\mu _\nu}^{(1)}+\tensor{R}{^\lambda _\rho _\mu _\nu}^{(2)} \right)\nonumber\\
&&- \frac{1}{2} \tensor{K}{^\lambda _\sigma _\mu}\tensor{K}{^\sigma_\rho _\nu}+\frac{1}{2} \tensor{K}{^\lambda _\sigma _\nu} \tensor{K}{^\sigma _\rho _\mu},
\end{eqnarray}
where $K$ is given by
\begin{equation}
    \tensor{K}{^\lambda _\mu _\nu}=\tensor{N}{^\lambda _\mu _\nu}^{(1)}- \tensor{N}{^\lambda _\mu _\nu} ^{(2)}.
\end{equation}

The mutual Ricci curvature reads
\begin{eqnarray}
\mathcal{R}_{\rho \nu}&=&\frac{1}{2} \left(\tensor{R}{_{\rho \nu}}^{(1)}+\tensor{R}{_{\rho \nu}}^{(2)} \right) - \frac{1}{2} \tensor{K}{^\lambda _\sigma _\lambda} \tensor{K}{^\sigma _\rho _\nu}\nonumber\\
    &&+\frac{1}{2} \tensor{K}{^\lambda _\sigma _\nu} \tensor{K}{^\sigma _\rho _\lambda}.
\end{eqnarray}

By transvecting with $g^{\rho \nu}$ we obtain the mutual Ricci scalar
\begin{equation}
    \mathcal{R}=\frac{1}{2} \left(R^{(1)}+R^{(2)} \right) - \frac{1}{2} \tensor{K}{^\lambda _\sigma _\lambda} \tensor{K}{^\sigma _\rho ^\rho} + \frac{1}{2} \tensor{K}{^\lambda _\sigma ^\rho} \tensor{K}{^\sigma _\rho _\lambda}.
\end{equation}

In contrast to the metric-affine approach taken in \cite{Iosifidis2023}, we fix the two connections and study the cosmological implications of the theory.

{We consider the biconnection theory, where $\nabla^{(1)}$ and $\nabla^{(2)}$ are the Schrödinger connection and its dual connection, respectively. This choice is well-motivated physically, as they both preserve the lengths under autoparallel transport, despite having nonmetricity.}

We introduce the notations
\begin{equation}
\begin{aligned}
    \tensor{R}{_{\mu \nu}}^{(1)}&:=R_{\mu \nu},\; \; R^{(1)}:=R,\\
    \tensor{R}{_\mu _\nu}^{(2)}&:=R_{\mu \nu}^{*},\; \;   R^{(2)}:=R^*,
\end{aligned}
\end{equation}
where the objects denoted with $*$ refer to the dual Schr\"{o}dinger connection. Similarly, we define the mutual difference tensor
\begin{equation}
    \tensor{D}{_\rho _\nu}:=-\frac{1}{2} \tensor{K}{^\lambda_\sigma _\lambda} \tensor{K}{^\sigma _\rho _\nu} + \frac{1}{2} \tensor{K}{^\lambda _\sigma _\nu} \tensor{K}{^\sigma _\rho _\lambda},
\end{equation}
and its contraction, the mutual difference scalar
\begin{equation}
    D:=-\frac{1}{2} \tensor{K}{^\lambda _\sigma _\lambda} \tensor{K}{^\sigma _\rho ^\rho} + \frac{1}{2} \tensor{K}{^\lambda _\sigma _\nu} \tensor{K}{^\sigma ^\nu _\lambda}.
\end{equation}

A lengthy algebraic computation detailed in Appendix~\ref{appendixB} gives:
\begin{enumerate}
    \item[$(i)$] Ricci tensor of the Schr\"{o}dinger connection
    \bea
       R_{\rho \nu}=\overset{\circ}{R}_{\rho \nu} &-& g_{\rho \nu} \overset{\circ}{\nabla}_{\alpha} \pi^{\alpha} + \frac{1}{2} \overset{\circ}{\nabla}_{\rho} \pi_{\nu} - \overset{\circ}{\nabla}_{\nu} \pi_{\rho}\nonumber\\
        &-&\frac{1}{2} \pi_{\sigma} \pi^{\sigma} g_{\rho \nu} - \frac{1}{4} \pi_\rho \pi_\nu.
    \eea
    \item[$(ii)$] Ricci tensor of the dual
    \bea
        R^{*}_{\rho \nu}=\overset{\circ}{R}_{\rho \nu} &-&\frac{1}{2} g_{\rho \nu} \overset{\circ}{\nabla}_{\alpha} \pi^{\alpha} + \overset{\circ}{\nabla}_{\rho} \pi_{\nu} + \overset{\circ}{\nabla}_{\nu} \pi_{\rho}\nonumber\\
        &+& \pi_\sigma \pi^\sigma g_{\rho \nu} + \frac{1}{2} \pi_\rho \pi_\nu.
    \eea
    \item[$(iii)$] Mutual difference tensor
    \begin{equation}
        D_{\rho \nu}=-\frac{1}{4} \pi_\rho \pi_\nu -\frac{1}{2} \pi_\sigma \pi^\sigma g_{\rho \nu}.
    \end{equation}
\end{enumerate}

From the above results, we easily obtain the mutual Ricci tensor
\bea\label{mutualcurvature}
\mathcal{R}_{\rho \nu}
    =\overset{\circ}{R}_{\rho \nu} &-& \frac{3}{4} g_{\rho \nu} \overset{\circ}{\nabla}_{\alpha} \pi^{\alpha} + \frac{3}{4} \overset{\circ}{\nabla}_{\rho} \pi_{\nu}\nonumber\\
    &-& \frac{1}{4} \pi^\sigma \pi_\sigma g_{\rho \nu} - \frac{1}{8} \pi_\rho \pi_\nu,
\eea
and its contraction, the mutual Ricci scalar
\begin{equation}\label{mutualcurvaturescalar}
    \mathcal{R}=\overset{\circ}{R} - \frac{9}{4} \overset{\circ}{\nabla}_{\alpha} \pi^{\alpha} - \frac{9}{8} \pi_\alpha \pi^\alpha.
\end{equation}

\subsection{The gravitational field equations}

By analogy with Einstein gravity, we postulate that the gravitational field equations are
\begin{equation}\label{fieldequations}
    \mathcal{R}_{(\rho \nu)} - \frac{1}{2} g_{\rho \nu} \mathcal{R}= 8\pi T_{\rho \nu}.
\end{equation}

{Let us mention that  the mutual curvature also has an antisymmetric part, namely
\begin{equation}
    \mathcal{R}_{[\rho \nu]}=\frac{1}{2} \left(\mathcal{R}_{\rho \nu} - \mathcal{R}_{\nu \rho} \right)=\frac{3}{8} \left(\overset{\circ}{\nabla}_{\rho} \pi_\nu - \overset{\circ}{\nabla}_{\nu} \pi_{\rho} \right)=\frac{3}{8} F_{\rho \nu},
\end{equation}
which can be though of as a Faraday-type tensor for the vector field $\pi$
\begin{equation}
    F_{\rho \nu}:=\overset{\circ}{\nabla}_{\rho} \pi_{\nu} - \overset{\circ}{\nabla}_{\nu} \pi_{\rho}.
\end{equation}
In our study, we set $\mathcal{R}_{[\rho \nu]} =0$, since an antisymmetric part in the Einstein equations is usually related to peculiar matter. Such matter types appear most often when geometry-matter coupling is considered. However, as we mentioned, we work in the metric formalism, without considering geometry-matter coupling. Moreover, the focus of this paper is on cosmological implications. In a FLRW-type metric, the vector $\pi$ has only a temporal component, thus the condition $\mathcal{R}_{[\rho \nu]}=0$ is automatically satisfied.}}

Using equations \eqref{mutualcurvature} and \eqref{mutualcurvaturescalar}, the field equations can be rewritten as
\begin{eqnarray}\label{feq}
        \overset{\circ}{R}_{\rho \nu} &-& \frac{1}{2} g_{\rho \nu} \overset{\circ}{R} + \frac{3}{8} g_{\rho \nu} \overset{\circ}{\nabla}_{\alpha} \pi^{\alpha}+\frac{3}{8} \overset{\circ}{\nabla}_{\rho} \pi_{\nu} +\frac{3}{8} \overset{\circ}{\nabla}_{\nu} \pi_{\rho}\nonumber\\
        &+&\frac{5}{16} g_{\rho \nu} \pi_{\alpha} \pi^{\alpha} -\frac{1}{8} \pi_{\rho} \pi_{\nu}=8 \pi T_{\rho \nu}.
\end{eqnarray}

By contracting equation ~(\ref{feq}) we obtain
\be\label{43}
-\overset{\circ}{R}+\frac{9}{4}\overset{\circ}{\nabla}_\alpha \pi ^\alpha +\frac{9}{8}\pi _\alpha \pi^\alpha =8\pi T.
\ee

Hence, we can reformulate the gravitational field equations in the form
\bea
\overset{\circ}{R}_{\rho \nu}&=&8\pi \left(T_{\rho \nu}-\frac{1}{2}g_{\rho \nu}T\right)+\frac{1}{4}\left(3\overset{\circ}{\nabla}_\alpha \pi ^\alpha+\pi _\alpha \pi^\alpha\right)g_{\rho \nu}\nonumber\\
&&-\frac{3}{8}\overset{\circ}{\nabla}_\rho \pi_\nu-\frac{3}{8}\overset{\circ}{\nabla}_\nu \pi_\rho+\frac{1}{8}\pi _\rho \pi_\nu.
\eea

{It is interesting to compare the field equations of the biconnection theory, Eqs.~(\ref{feq}), based on the averaging of two connections, with the field equations (\ref{eqWS}) of the Weyl-Schr\"{o}dinger theory. As expected, the field equations of the two theories have a very similar structure, containing a symmetric dependence on the covariant derivatives of the vector fields, on the divergence of the vector fields, and on the square of the vectors. However, it is important to point out the geometric differences between the two approaches. While in the Weyl-Schr\"{o}dinger approach the basic geometric quantity is the nonmetricity, in the biconnection theory the nonmetricity plays a dual role, also determining the torsion tensor.   Another similarity for both theories is that they do not introduce any new physical parameters, or coupling constants. However, differences do appear in the signs of the various similar terms, as well as in the numerical values of the coefficients multiplying the contributions of the nonmetricity. Even that the qualitative behaviour of the two models is similar, the differences in the values of the coefficients leads to quantitative differences in the predictions of the two theories.   
}      

\subsubsection{Divergence of the matter energy-momentum tensor}

By taking the covariant divergence of equation ~(\ref{feq}) with respect to the Riemannian divergence operator $\overset{\circ}{\nabla}$, and by taking into account that the divergence of the Einstein tensor identically vanishes, we obtain for the divergence of the matter energy-momentum tensor the expression
\bea\label{dem}
\hspace{-0.5cm}8\pi \overset{\circ}{\nabla}_\rho \tensor{T}{^\rho _\nu}&=&\frac{3}{8}\overset{\circ}{\Box}\pi _\nu+\frac{3}{8}\left(\overset{\circ}{\nabla}_\nu \overset{\circ}{\nabla}_\rho +\overset{\circ}{\nabla}_\rho \overset{\circ}{\nabla}_\nu\right)\pi ^{\rho}\nonumber\\
\hspace{-0.5cm}&&+\frac{5}{16}\overset{\circ}{\nabla}_\nu \left(\pi _\rho\pi ^\rho\right)
-\frac{1}{8}\overset{\circ}{\nabla}_\rho \left(\pi _\nu \pi^\rho\right)\equiv A_\nu,
\eea
where $\overset{\circ}{\Box}=\overset{\circ}{\nabla}_\rho \overset{\circ}{\nabla} \tensor{}{^\rho}$. Hence, generally, in the present biconnection gravitational theory the matter energy-momentum tensor does not vanish identically.

{The set of the field equations (\ref{feq}), as well as their consequence Eq.~(\ref{dem}) do not form a closed system of equations, due to the lack of an equation of motion for the vector field. An equation of motion for $\pi_\mu$ can be obtained by imposing by hand the requirement of the conservation of the matter energy-momentum tensor,} $\overset{\circ}{\nabla}_\rho \tensor{T}{^\rho _\nu}\equiv 0$. {In this case we obtain for the Weyl vector the evolution equation}
\begin{eqnarray}\label{46}
&&\overset{\circ}{\Box}\pi _\nu+\left(\overset{\circ}{\nabla}_\nu \overset{\circ}{\nabla}_\rho +\overset{\circ}{\nabla}_\rho \overset{\circ}{\nabla}_\nu\right)\pi ^{\rho}\nonumber\\
&&+\frac{10}{3}\overset{\circ}{\nabla}_\nu \left(\pi _\rho\pi ^\rho\right)
-\frac{1}{3}\overset{\circ}{\nabla}_\rho \left(\pi _\nu \pi^\rho\right)=0,
\end{eqnarray} 
{which, once an appropriate set of initial and/or boundary conditions are given, uniquely determines the vector} $\pi _\mu$ {in the vacuum case. However, this situation arises only because the Schr\"{o}dinger connection is described in terms of a single extra vector field.  If the connection would contain more than one fields} $\pi_\mu^{(a)}$, $a=1,2,...n$, {even in the vacuum case the condition of the vanishing of the matter energy-momentum tensor can not uniquely determine an equation of motion for the fields. On the other hand, in the presence of matter, to close the system of the field equations one must impose a supplementary condition (equation of state) for both the baryonic and geometric matter components. Such an equation of state (geometrical or physical) closes the system of field equations, and allows to obtain fully consistent gravitational models.}

{In the present approach we have postulated the gravitational field equations as having the same form as in standard general relativity. A systematic approach for the construction of biconnection gravitational theories can be obtained via a variational principle. If such a principle could be formulated, it would allow to obtain sets of independent equations of motion for the extra vector fields that could be included in the structure of the connections. However, the mutual Ricci scalar (\ref{mutualcurvaturescalar}), even though it represents a generalization of the Riemannian Ricci scalar, cannot be used as a Lagrangian density to obtain the field equations (\ref{feq}). Hence, obtaining an action principle for the field equations of the present biconnection gravity theory is still an open problem.             
}

\paragraph{Equation of motion of massive particles.} The equation of motion for a massive test particle can be found from equation ~(\ref{dem}). We adopt for the matter source a perfect fluid, which is described by two thermodynamic quantities only, the energy density $\rho$, and the thermodynamic pressure $p$.  The energy-momentum tensor of the fluid is then given by
\be\label{perfectfluid}
T_{\mu\nu}=(\rho+p)u_\mu u_\nu+pg_{\mu\nu},
\ee
where $u^{\mu}$ is the four-velocity of the particle, normalized according to $u^\mu u_\mu=-1$. We also introduce the projection operator $h_{\mu \nu }$, defined according to
$h_{\mu \nu }=g_{\mu \nu }+u_{\rho}u_{\nu}$, and which has the property $h_{\mu \lambda}u^\mu=0$. By taking the divergence of equation ~(\ref{perfectfluid}), we obtain
\bea
\hspace{-0.5cm}\overset{\circ}{\nabla}_\mu T^{\mu\nu}&=&h^{\mu\nu}\overset{\circ}{\nabla}_\mu p + u^\nu u_\mu\overset{\circ}{\nabla}\tensor{}{^\mu} \rho\nonumber\\
\hspace{-0.5cm}&&+(\rho+p)\big(u^\nu\overset{\circ}{\nabla}_\mu u^\mu+u^\mu\overset{\circ}{\nabla}_\mu u^\nu\big)=\frac{1}{8\pi}A^\nu.
\eea

We multiply now the above equation with $h_\nu^\lambda$ to find
\be
h_\nu^\lambda\overset{\circ}{\nabla}_\mu T^{\mu\nu}=(\rho+p)u^\mu\overset{\circ}{\nabla}_\mu u^\lambda+h^{\nu\lambda}\overset{\circ}{\nabla}_\nu p=\frac{1}{8\pi}h_\nu^\lambda A^\nu,
\ee
where we have used the identity $u_\mu\overset{\circ}{\nabla}_\nu u^\mu=0$.
Hence the equation of motion for a massive test particle in length-preserving biconnection gravity takes the form
\be\label{eq404}
\frac{d^2x^\lambda}{ds^2}+\overset{\circ}{\Gamma}\tensor{}{^\lambda_{\mu\nu}}u^\mu u^\nu=\frac{1}{(\rho +p)}h^{\nu \lambda} \left(\frac{1}{8\pi}A_\nu-\overset{\circ}{\nabla}_\nu p\right)= f^\lambda,
\ee
where we have used the definition of the covariant derivative to obtain $u^\mu\overset{\circ}{\nabla}_\mu u^\lambda$ in the left hand side of equation ~(\ref{eq404}). 

Moreover, by $\overset{\circ}{\Gamma}\tensor{}{^\lambda_{\mu\nu}}$ we have denoted the Levi-Civita connection associated to the metric. Hence, the motion of the massive particles in the present theory is non-geodesic, and an extra force $f^\lambda$ is generated.   The extra-force is perpendicular to the four-velocity, and it satisfies the condition $f^{\lambda }u_{\lambda }=0$. If the torsion vector vanishes,  the extra-force takes the form $f^{\lambda }=-h^{\lambda \nu}\nabla _{\nu }p/\left(\rho +p\right)$, corresponding to the standard general relativistic fluid motion.

\subsubsection{The Newtonian limit}

We assume that one can formally represent $f^\lambda$ as
\be\label{eq407}
f^{\lambda} = (g^{\nu \lambda}+u^{\nu}u^{\lambda})\overset{\circ}{\nabla}_{\nu} \ln \sqrt{Q}=h^{\nu \lambda}\overset{\circ}{\nabla}_{\nu} \ln \sqrt{Q},
\ee
or
\be
\frac{1}{(\rho +p)} \left(\frac{1}{8\pi}A_\nu-\overset{\circ}{\nabla}_\nu p\right)=\overset{\circ}{\nabla}_{\nu} \ln \sqrt{Q},
\ee
where we have introduced the dimensionless function $Q$ to describe the effects of the extra force. We assume that $Q$ is not an
explicit function of $u^{\mu}$. The equation of motion ~\eqref{eq404} can be obtained from the variational principle
\be\label{eq408}
S_p=\int L_p \; ds=\int\sqrt{Q}\sqrt{g_{\mu\nu}u^{\mu}u^{\nu}} \; ds,
\ee
where $S_p$ and $L_p$ are the action and the  Lagrangian density of the test particle.
In the limit $\sqrt{Q}\rightarrow1$, we reobtain the variational principle for the motion of the  standard general relativistic particles. The equivalence between equation ~(\ref{eq404}) and the variational principle (\ref{eq408}) can be proven by writing down the Lagrange equations
corresponding to the action~(\ref{eq408}),
\begin{equation}
\frac{d}{ds}\left( \frac{\partial L_{p}}{\partial u^{\lambda }}\right)
 - \frac{\partial L_{p}}{\partial x^{\lambda }}=0\, .
\end{equation}

Then, we successively obtain
\begin{equation}
\frac{\partial L_{p}}{\partial u^{\lambda }}=\sqrt{Q}u_{\lambda }
\end{equation}
and
\begin{equation}
\frac{\partial L_{p}}{\partial x^{\lambda }}
=\frac{1}{2} \sqrt{Q}g_{\mu \nu,\lambda }u^{\mu }u^{\nu }
+\frac{ 1}{2} \frac{Q_{,\lambda }}{Q}\, ,
\end{equation}
respectively. Finally,  a simple calculation gives the equations of motion of the
particle as
\begin{equation}
\frac{d^{2}x^{\mu }}{ds^{2}}+\tensor{\Gamma}{^\mu_\nu_\lambda}u^{\nu }u^{\lambda
}+\left( u^{\mu }u^{\nu }+g^{\mu \nu }\right) \overset{\circ}{\nabla} _{\nu }\ln \sqrt{Q}=0.
\end{equation}

When $\sqrt{Q}\rightarrow 1$ the standard general relativistic equation for geodesic motion are reobtained.

The Newtonian limit of the theory can be studied by using the variational principle equation for ~\eqref{eq408}. In the weak gravitational field limit, the interval $ds$ for a dust fluid, with $p=0$, in motion in the gravitational field is given by
\be\label{eq415}
ds\approx \sqrt{1+2\phi-\vec{v}^2} \; dt\approx \left(1+\phi-\frac{\vec{v}^2}{2}\right) \, dt,
\ee
where $\phi$ is the Newtonian potential and $\vec{v}$ is the three-dimensional
velocity of the fluid. We also represent $\sqrt{Q}=1+U$, $\ln \sqrt{Q}=\ln (1+U)\approx U$, thus obtaining
\be
\overset{\circ}{\nabla}_\nu U\approx \frac{1}{8\pi \rho}A_\nu.
\ee

In the first order of approximation the equation of motion of the fluid is obtained from the variational principle
\be\label{eq417}
\delta\int \left[1+U+\phi-\frac{\vec{v}^2}{2}\right]dt=0.
\ee

By writing down the equation of motion corresponding to the variational principle (\ref{eq417}), we obtain the total acceleration $\vec{a}$ as given by
\be\label{eq418}
\vec{a}=-\overset{\circ}{\nabla} \phi-\overset{\circ}{\nabla} U=-\overset{\circ}{\nabla} \phi-\frac{1}{8\pi \rho}\vec{A}=\vec{a}_N+\vec{a}_E,
\ee
where $\vec{a}_N=-\overset{\circ}{\nabla} \phi$ is the Newtonian acceleration, and the
extra acceleration, induced by the presence of the torsion, is
\be\label{accf}
\vec{a}_E=-\overset{\circ}{\nabla} U\approx -\frac{1}{8\pi \rho}\vec{A},
\ee
and 
\be
a_E^2=\vec{a}_E\cdot \vec{a}_E=\frac{1}{64\pi^2 \rho ^2}\vec{A}^2,
\ee
respectively.

The acceleration given by equation ~(\ref{eq418}) is due to the presence of the torsion in the biconnection gravitational model.
Since we have assumed that the fluid is pressureless, there is no hydrodynamical acceleration $\vec{a}_p$ term in the
expression of the total acceleration. Such an acceleration term
does exist in the general case. It is interesting to note that the extra-acceleration $\vec{a}_e$ depends not only on the properties of the torsion vector, but also on the density of the fluid.

\section{Cosmological applications}\label{sect3}

In the present Section we investigate the cosmological implications of the generalized gravitational field equations (\ref{feq}). As a first step in our study we obtain the generalized Friedmann equations, corresponding to a flat Friedmann-Lemaitre-Robertson-Walker metric. The existence of a de Sitter type solution is also investigated. Several cosmological models, corresponding to different choices of an effective equation of state for the geometric energy and pressure components are also studied. A comparison with the observational data is performed as well.

\subsection{The generalized Friedmann equations}

We consider a flat FLRW metric
\begin{equation}
    ds^2=-dt^2+a(t)^2 \delta_{ij} dx^i dx^j.
\end{equation}

For matter, we take a perfect fluid with the energy-momentum tensor given by Eq.~(\ref{perfectfluid}).
Due to the requirement of homogeneity and isotropy of the Universe, the field $\pi$ can have only a temporal component
\begin{equation}
    \pi^{\mu}=(\psi(t),0,0,0) \iff \pi_\mu=(-\psi(t),0,0,0).
\end{equation}

With these assumptions, a calculation detailed in Appendix~\ref{appendixC} yields the following Friedmann equations
\begin{equation}\label{57}
    3H^2 =8 \pi \rho+\frac{9}{8}\dot{\psi}+\frac{9}{8}H\psi -\frac{3}{16}\psi ^2=8\pi \left(\rho+\rho _{eff}\right),
\end{equation}
\begin{equation}\label{58}
   2 \dot H + 3 H^2= -8 \pi p + \frac{3}{8} \dot \psi + \frac{15}{8} H \psi - \frac{5}{16} \psi^2=-8\pi \left(p+p_{eff}\right),
   \end{equation}
   where we have denoted
   \be
   \rho_{eff}=\frac{1}{8\pi}\left(\frac{9}{8}\dot{\psi}+\frac{9}{8}H\psi -\frac{3}{16}\psi ^2\right),
   \ee
   and
   \be
   p_{eff}=\frac{1}{8\pi}\left(- \frac{3}{8} \dot \psi - \frac{15}{8} H \psi +\frac{5}{16} \psi^2\right),
   \ee
respectively. From the generalized Friedmann equations (\ref{57}) and (\ref{58}) we obtain the conservation equation
\be
\frac{d}{dt}\left[a^3\left(\rho +\rho_{eff}\right)\right]+\left(p+p_{eff}\right)\frac{d}{dt}a^3=0,
\ee
or, equivalently,
\be
\dot{\rho}+3H(\rho+p)+\dot{\rho}_{eff}+3H\left(\rho_{eff}+p_{eff}\right)=0.
\ee

To facilitate the comparison of the theoretical predictions with the observational data we will use as an independent variable the redshift $z$, defined as $1+z=1/a$. Then we can replace the time variable by $z$ according to the relation
\be
\frac{d}{dt}=-(1+z)H(z)\frac{d}{dz}.
\ee

\paragraph{Dimensionless and redshift representation.} To simplify the mathematical formalism we introduce a set of dimensionless variables $(\tau, h, r, P, \Psi)$, defined according to
\be
\tau =H_0t, H=H_0h, \rho =\rho _cr, p=\frac{1}{3}\rho _c P, \psi =H_0\Psi,
\ee
where $H_0$ denotes the present day value of the Hubble function and $\rho_c=\frac{3 H_0^2}{8 \pi}$. Hence the Friedmann equations take the dimensionless form
\be
h^2=r+\frac{3}{8}\frac{d\Psi}{d\tau}+\frac{3}{8}h\Psi-\frac{1}{16} \Psi^2,
\ee
\be
2\frac{dh}{d\tau}+3h^2=-P+\frac{3}{8}\frac{d\Psi}{d\tau}+\frac{15}{8}h\Psi-\frac{5}{16}\Psi^2.
\ee

In the redshift space we obtain the evolution equations
\be \label{F1redshift}
h^2(z)=r(z)-\frac{3}{8}(1+z)h(z)\frac{d\Psi}{dz}+\frac{3}{8}h(z)\Psi(z)-\frac{1}{16} \Psi^2(z),
\ee
\bea
-2(1+z)h(z)\frac{dh(z)}{dz}+3 h^2(z)&=&-P(z)\nonumber\\
&&-\frac{3}{8}(1+z)h(z)\frac{d\Psi (z)}{dz}\nonumber\\
&&+\frac{15}{8}h(z)\Psi (z)-\frac{5}{16}\Psi^2 (z).\nonumber\\
\eea

\subsubsection{The de Sitter solution}

We consider now the de Sitter type solution of the generalized Friedmann equations, corresponding to $H=H_0={\rm constant}$. By assuming a dust Universe with $p=0$, the second Friedman equation (\ref{58}) gives for $\psi$ the evolution equation
\be
\dot{\psi}+5H_0\psi-\frac{5}{6}\psi ^2-8H_0^2=0,
\ee
with the general solution given by
\begin{equation}
\psi (t)=3H_{0}\left\{ 1+\frac{1}{\sqrt{15}}\tan \left[ \frac{\sqrt{15}}{6}%
\left( H_{0}t+\alpha \right) \right] \right\} ,
\end{equation}
where we have used the initial condition $\psi (0)=\psi_0$, and we have introduced the notation $\alpha =\left( 6/\sqrt{15}\right) \tan ^{-1}\left[
\sqrt{15}\left( \psi _{0}-3\sqrt{15}H_{0}\right) /3H_{0}\right] /\left(\sqrt{15}H_0\right)$. The evolution of the matter energy density is obtained as
 \be
 8\pi \rho (t)=\frac{3}{20} H_0^2 \left\{8-3 \sec ^2\left[\frac{1}{2} \sqrt{\frac{5}{3}} \left(
   H_0 t+\alpha\right)\right]\right\}.
 \ee

 It is interesting to note that the matter energy density is a periodic function. However, it reaches the zero value after a finite time interval $t_f=(1.41-\alpha)/H_0$, which represents the end of the de Sitter phase for the present biconnection model.

\subsection{Specific cosmological models}

In the present subsection we will consider several examples of specific cosmological models in the framework of length-preserving biconnection gravity. For each case we will also consider a detailed comparison with the observational data, as well as with the predictions of the standard $\Lambda$CDM model.
In the following we will restrict our analysis to the case of a Universe filled with a pressureless dust, with $p=0$.

\subsubsection{Conservative cosmological (CC) model}\label{ss:model1}

As a first example of a cosmological model we consider the case in which the matter energy density is conserved, thus satisfying the equation
\be\label{81}
\dot{\rho}+3H(\rho+p)=0.
\ee

Therefore, the effective geometric energy density is also conserved, and we have
\be
\dot{\rho}_{eff}+3H\left(\rho_{eff}+p_{eff}\right)=0,
\ee
giving for the torsion vector the evolution equation
\be\label{tors}
\ddot{\psi}+\dot{H}\psi+3H\dot{\psi}-\frac{1}{3}\psi \dot{\psi}-2H^2\psi+\frac{1}{3}H\psi^2=0.
\ee

{The conservation equation (\ref{81}) gives for the evolution of the matter density the relation $\rho= \rho_0/a^3$, where $\rho _0$ is an integration constant. By substituting this relation for $\rho$ into the generalized Friedmann equation (\ref{57}), and considering it together with Eq.~(\ref{58}),  one obtains a system of first order differential equations.  However, in the following we will use in the numerical investigations the second order differential equation (\ref{tors}) for the evolution of torsion, since it offers a clear physical and cosmological picture on the time variation of $\psi$, and on its dependence on the Hubble function, and on its derivative.    
} 

 We also introduce a new variable $u=d\Psi/d\tau$. Thus, the system of equations describing the conservative cosmological evolution in the biconnection gravity theory takes the form
 \be\label{m11}
 -(1+z)h(z)\frac{d\Psi (z)}{dz}=u(z),
 \ee
 \bea\label{m12}
&& -(1+z)h(z)\frac{du(z)}{dz}-(1+z)h(z)\frac{dh(z)}{dz}\Psi (z)\nonumber\\
&&+3h(z)u(z)  -\frac{1}{3}\Psi (z)u(z)\nonumber\\
&&-2h^2(z)\Psi (z)+\frac{1}{3}h(z)\Psi ^2(z)=0,
  \eea
\bea\label{m13}
-2(1+z)h(z)\frac{dh(z)}{dz}+3 h^2(z)&=&
\frac{3}{8}u(z)+\frac{15}{8}h(z)\Psi (z)\nonumber\\
&&-\frac{5}{16}\Psi^2 (z).
\eea

The system of equations (\ref{m11})-(\ref{m13}) must be integrated with the initial conditions $h(0)=1$, $\Psi (0)=\Psi_0$, and $u(0)=u_0$, respectively.

{From the first Friedmann equation \eqref{F1redshift}, we  obtain the evolution of the matter density
\begin{equation}\label{eq:req1}
    r(z)=h^2(z)-\frac{3}{8}u(z) -\frac{3}{8}h(z) \Psi(z) + \frac{1}{16} \Psi^2(z).
\end{equation}

Hence, the present day matter density $r(0)$ is determined by the initial values $u_0$ and $\Psi_0$ as
\begin{equation}
    r(0)=1-\frac{3}{8} u_0 - \frac{3}{8} \Psi_0 + \frac{1}{16}\Psi_0^2.
\end{equation}

\subsubsection{Linear equation of state cosmological (LESC) model}

As a second cosmological model in length-preserving biconnection gravity with Schrödinger connections we consider the case in which the dark energy effective pressure and density are related by a linear equation of state, given by
\begin{equation}
    p_{eff}(z)=\omega_0 \rho_{eff}(z).
\end{equation}

In this model, the equations describing the evolution of the Universe are given by
\begin{equation}\label{m21}
\begin{aligned}
(1+z)h(z)\frac{d\Psi (z)}{dz}= \frac{\left[3 \omega_0 +5\right] \Psi (z) \left[6 h(z)-\Psi (z)\right]}{6 \left[3\omega_0
   +1\right]},
\end{aligned}
\end{equation}
\bea\label{m22}
-2(1+z)h(z)\frac{dh(z)}{dz}+3 h^2(z)&=&
-\frac{3}{8}(1+z)h(z)\frac{d\Psi (z)}{dz}\nonumber\\
&&+\frac{15}{8}h(z)\Psi (z)-\frac{5}{16}\Psi^2 (z).\nonumber\\
\eea

The system of equations ~(\ref{m21}) and (\ref{m22}) must be solved with the initial conditions $h(0)=1$, and $\Psi (0)=\Psi _0$.

{After solving the system, we obtain the matter density using the closure relation
\begin{equation}
    r(z)=h^2(z)+\frac{3}{8}(1+z)h(z) \frac{d \Psi(z)}{dz}-\frac{3}{8}h(z) \Psi(z)+\frac{1}{16} \Psi^2(z).
\end{equation}

\section{Cosmological tests of the CC and LESC models}\label{sect4}

{In this Section, we compare the predictions of the proposed cosmological models with those of the standard $\Lambda$CDM model, as well as with a small set of observational data for the Hubble function. To recap, the Hubble function for the $\Lambda$CDM model is expressed as
\begin{equation}
    H(z)=H_0 \sqrt{\Omega_M(1+z)^3+\Omega_\Lambda},
\end{equation}
where $H_0$ is the present day value of the Hubble function, $\Omega_m$ is the current  matter density, and $\Omega_\Lambda$ represents the dark energy density. These two parameters satisfy the constraint
\begin{equation}
    \Omega_M+\Omega_\Lambda=1.
\end{equation}

\subsection{Parameter estimation}

To begin our comparison, we determine the optimal fitting values for the $\Lambda$CDM model and the two other models, respectively. These are found by performing a Likelihood analysis, using observational data of the Hubble function within the redshift range $z \in (0.07,2.36)$ as provided in \cite{Bouali_2023}.

The key ingredient in the statistical analysis is the likelihood function
\begin{equation}
    L = L_0 e^{-\chi^2/2},
\end{equation}
where $L_0$ is a normalization constant, and $\chi^2$ is the chi-squared statistic. This function is also known as the negative logarithmic likelihood function \cite{wackerly2008mathematical}, which is crucial for practical purposes since likelihood values can be very small. The chi-squared statistic is defined as
\begin{equation}
    \chi^2 = \sum_{i} \left( \frac{O_i - T_i}{\sigma_i} \right)^2.
\end{equation}

Here $i$ ranges over the data points, $O_i$ are the values obtained from the observational data, $T_i$ are the values predicted by the theory and $\sigma_i$ are the errors associated with the $i$-th data point.

The best-fit values of the parameters are determined by maximizing the likelihood function, which is equivalent to minimizing the chi-squared statistic. For the $\Lambda$CDM model they are given in Table \ref{tab:confidence_intervals3}, while for conservative cosmological (CC) model and the linear equation of state cosmological model (LESC) they are found in Table \ref{tab:confidence_intervals} and \ref{tab:confidence_intervals2}, respectively.

\begin{table}[h!]
    \centering
    \begin{tabular}{|c|c|c|}
        \hline
        & $H_0^{\Lambda}$ & $\Omega_M$ \\
        \hline
        $1\sigma$ & $70.01^{+0.446}_{-0.43}$ & $0.264^{+0.007}_{-0.006}$ \\
        \hline
        $2\sigma$ & $70.01^{+0.882}_{-0.866}$ & $0.264^{+0.014}_{-0.013}$ \\
        \hline
    \end{tabular}
    \caption{Optimal parameter values and their confidence intervals for the $\Lambda$CDM model.}
    \label{tab:confidence_intervals3}
\end{table}

\begin{table}[h!]
    \centering
    \begin{tabular}{|c|c|c|c|c|}
        \hline
        & $H_0$ & $\psi_0$ & $u_0$ \\
        \hline
        $1\sigma$ & $66.959^{+0.413}_{-0.425}$ & $1.097^{+0.018}_{-0.018}$ & $0.777^{+0.074}_{-0.081}$ \\
        \hline
        $2\sigma$ & $66.959^{+0.830}_{-0.842}$ & $1.097^{+0.036}_{-0.036}$ & $0.777^{+0.138}_{-0.155}$ \\
        \hline
    \end{tabular}
    \caption{Optimal parameter values and their confidence intervals for the conservative cosmological (CC) model.}
    \label{tab:confidence_intervals}
\end{table}

\begin{table}[h!]
    \centering
    \begin{tabular}{|c|c|c|c|c|}
        \hline
        & $H_0$ & $\psi_0$ & $\omega_0$ \\
        \hline
        $1\sigma$ & $66.224^{+0.410}_{-0.419}$ & $1.055^{+0.018}_{-0.018}$ & $-1.173^{+0.059}_{-0.059}$ \\
        \hline
        $2\sigma$ & $66.224^{+0.822}_{-0.831}$ & $1.055^{+0.036}_{-0.036}$ &  $-1.173^{+0.059}_{-0.059}$ \\
        \hline
    \end{tabular}
    \caption{Optimal parameter values and their confidence intervals for the linear equation of state cosmological (LESC) model.}
    \label{tab:confidence_intervals2}
\end{table}

The $1\sigma$ confidence intervals were determined using the distribution of
\begin{equation}
    \Delta \chi^2_i = \chi^2_i - \chi_o^2 \,
\end{equation}
where $\chi_o^2$ represents the negative logarithmic likelihood value at the optimal parameter, and $i$ spans a range of values around this optimal point.
The $\Delta \chi^2_i$ values follow a chi-squared distribution with one degree of freedom ($df = 1$), as we vary one parameter at a time \cite{wackerly2008mathematical}.
The critical value, which encapsulates $68.3\%$ of the area under the curve  (corresponding to one standard deviation), is approximately $\Delta \chi^2_i \simeq 1$.
Parameters were identified where this critical threshold is exceeded, both below and above the optimal value.

Similarly, the $2\sigma$ confidence interval is determined with a critical value that encompasses $95.4\%$ of the area under the curve, corresponding to two standard deviations, which is approximately $\Delta \chi^2_i \simeq 4$.
Additionally, the two-dimensional $\Delta \chi^2_{ij}$ surface can be examined, when two parameters are varied simultaneously. For $df=2$, the critical values are approximately $\Delta \chi^2_{ij}\simeq 2.3$ and $\Delta \chi^2_{ij} \simeq 6.2$, corresponding to $1\sigma$ and $2\sigma$ confidence levels, respectively. This approach allows for visual conclusions regarding the relational behaviour of parameter pairs relative to the optimum (see Fig.~\ref{cornerplot}).

\begin{figure*}[htbp]
\centering
\includegraphics[width=0.490\linewidth]{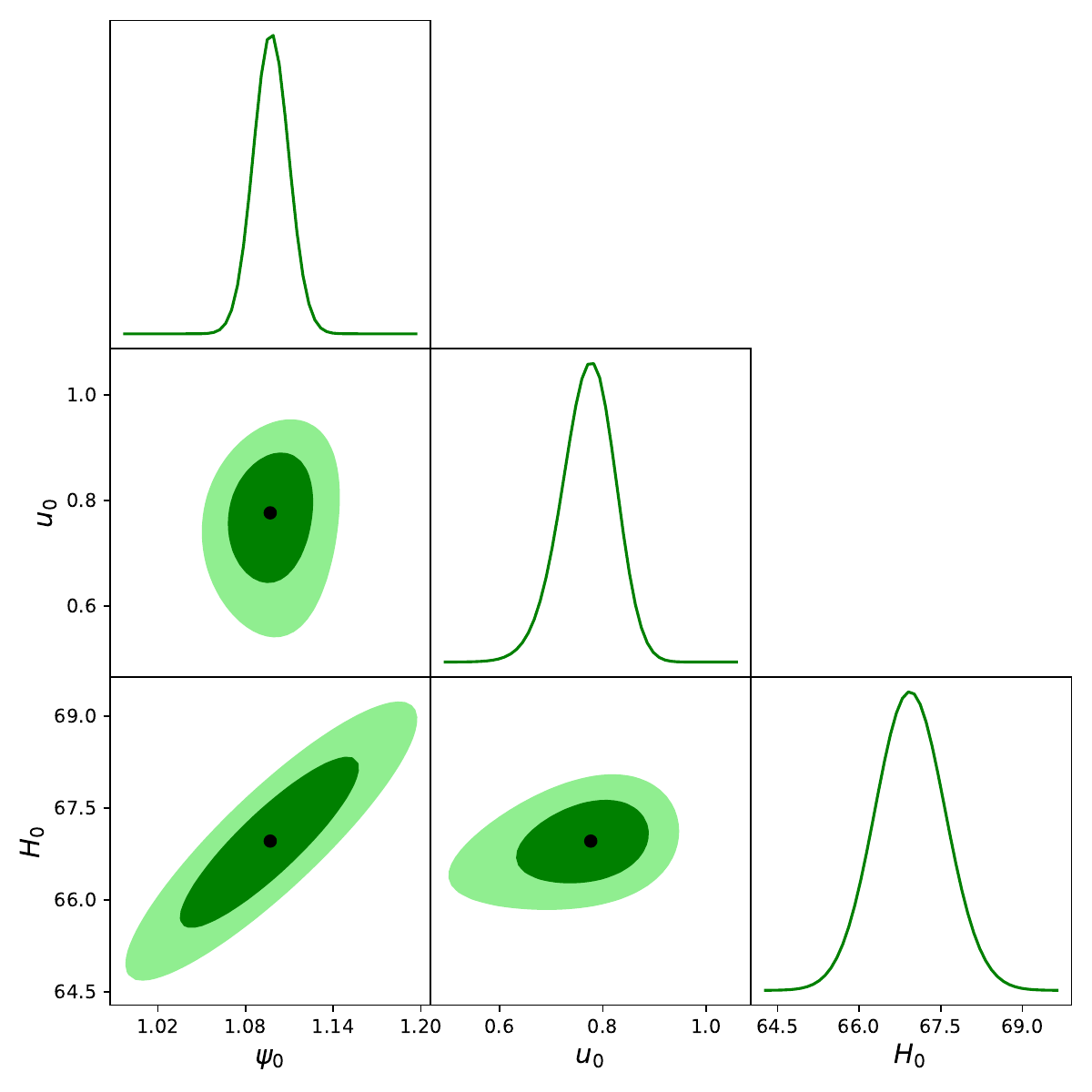} %
\includegraphics[width=0.490\linewidth]{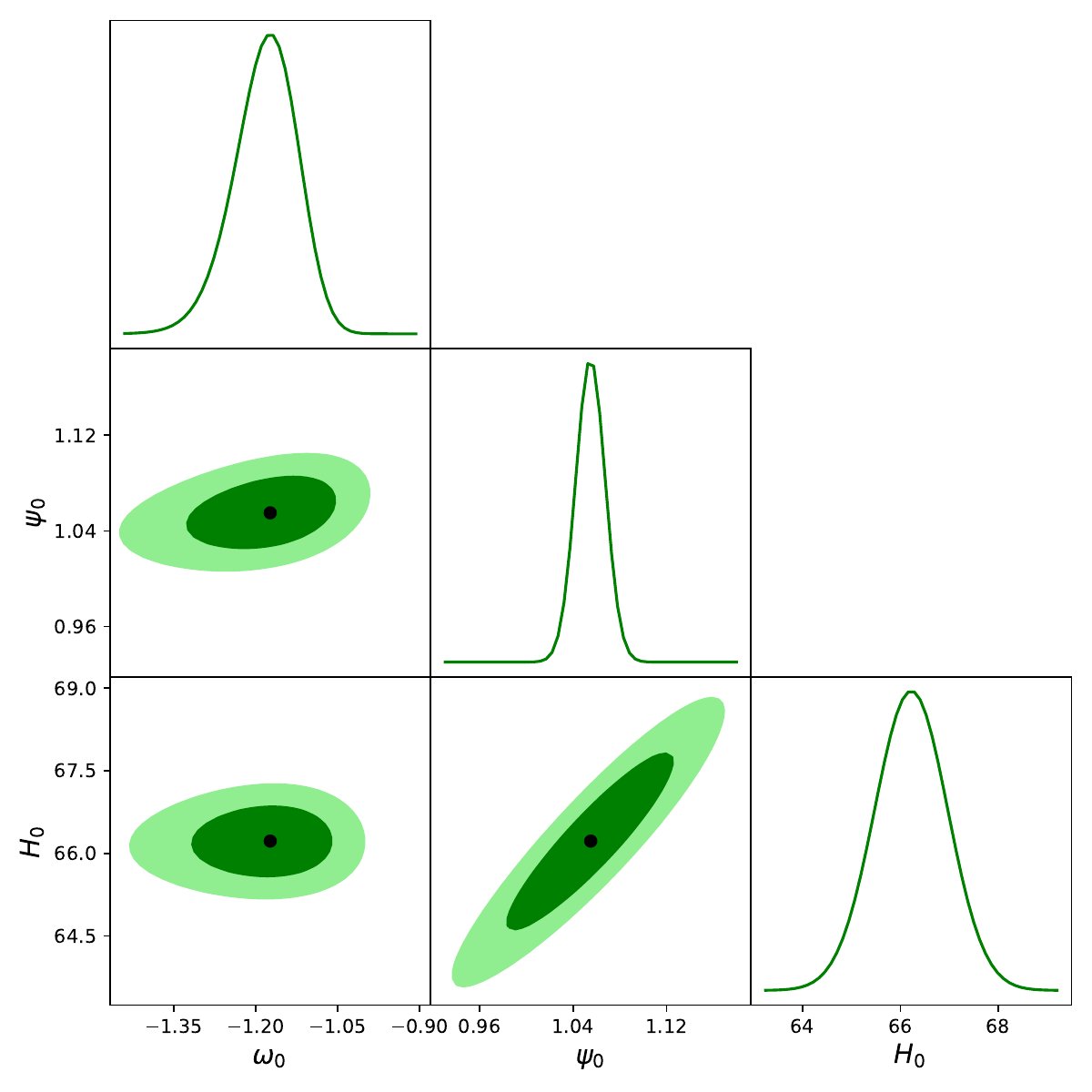}
\caption{Corner plot for the parameters $H_0, \Psi_0,u_0$ of the CC model (left panel) and for the parameters $H_0, \Psi_0, \omega_0$ of the LESC model, showing their $1 \sigma$ and $2 \sigma$ confidence intervals. Dark green regions indicate $1\sigma$ confidence, while light green regions represent $2\sigma$ confidence on the contour plots. Marginal distributions of the parameters are also displayed.}
\label{cornerplot}
\end{figure*}

To identify the best model that describes the observational data, we employ two information criteria: the corrected Akaike Information Criterion \cite{Akaike1974,Vrieze2012,RezaeiMalekjani2021, TanBiswas2012} ($AIC_c$) and the Bayesian Information Criterion ($BIC$) \cite{BurnhamAnderson2004a, BurnhamAnderson2004b}.

{In this study, we use $AIC_c$ instead of $AIC$ due to our small sample size, as $AIC$ tends to favor models with more parameters in such cases. $AIC_c$ addresses this issue by including a correction term for small datasets. It is defined as
\begin{equation}
    AIC_c= \chi^{2}_{min}+2k+\frac{2k^2+2k}{n-k-1},
\end{equation}
where $n$ is the number of data points and $k$ is the number of free parameters in the model. It is easily seen that the correction term vanishes as $n \to \infty$, making $AIC_c$ essentially identical to $AIC$ for large datasets.}

{The model with the smallest $AIC_c$ value is the most supported by the observational data, and is usually chosen to be the reference model. To evaluate how closely other models resemble the reference model, we compute the following quantity
\begin{equation}
    \Delta AIC_{c}=AIC_{c,model}- AIC_{c,reference}.
\end{equation}}

{
Depending on the $\Delta AIC_{c}$ value, models can be categorized as follows: well-supported by observations ($0<\Delta AIC_c<2$), moderately supported by observations ($4< \Delta AIC_c<7)$ or not supported by observations ($\Delta AIC_c >10)$.}

{The Bayesian Information Criterion ($BIC$) is another widely utilized method for model selection, rooted in Bayesian probability theory. Like $AIC_c$, $BIC$ evaluates models based on their fit to the data, but it imposes a heavier penalty for models with a larger number of parameters. Hence, it treats more severely the question of overfitting than $AIC_c$, often favoring simpler models. Formally, the $BIC$ is defined as
\begin{equation}
    BIC=\chi^{2}_{min}+k \ln(n).
\end{equation}}

{
In a similar fashion, the model with the lowest $BIC$ value is considered the best fit for the data and is selected as the reference model. The difference
\begin{equation}
    \Delta BIC=BIC_{model}-BIC_{reference}
\end{equation}
indicates the relative quality of the model compared to the reference. A model is considered weakly disfavored by the data if $0< \Delta BIC <2$, moderately disfavored if $2< \Delta BIC \leq 6$, and strongly disfavored if $\Delta BIC >6$.}

\begin{table*}[ht]
\begin{center}
\small 
\begin{tabular}{|>{\centering\arraybackslash}m{2.5cm}|>{\centering\arraybackslash}m{2cm}|>{\centering\arraybackslash}m{1cm}|>{\centering\arraybackslash}m{2cm}|>{\centering\arraybackslash}m{2cm}|>{\centering\arraybackslash}m{2cm}|>{\centering\arraybackslash}m{2cm}|}
\hline
Model & $\chi^{2}_{min}$ & $k$ & $AIC_c$ & $\Delta AIC_c$ & $BIC$ & $\Delta BIC$ \\
\hline
$\Lambda$CDM & 44.2156 & 3 & 48.4379 & 0 & 52.3018 & 0 \\
$LESC$ & 43.796 & 3 & 48.8324 & 0.3945 & 54.5087 & 2.2069 \\
$CC$ & 41.9618 & 3 & 48.4147 & -0.0262 & 54.0911 & 1.7893 \\
\hline
\end{tabular}
\caption{Summary of $\chi^{2}_{min}$, $AIC_c$, and $BIC$ for the three models.}
\label{table:models}
\end{center}
\end{table*}

{The $\chi^2$, $AIC_c$, and $BIC$ values of the three models are presented in Table \ref{table:models}. Based on the corrected Akaike Information Criterion, the conservative model best describes the data and should be considered the reference model. However, given our small data set and the common practice of using the $\Lambda$CDM model as the reference, we will adopt this approach. Consequently, the $\Delta AIC_c$ for the conservative model is negative. For the $LESC$ model, the $\Delta AIC_c$ is $0.3945$, indicating that this model is well-supported by the data. According to the $BIC$ comparison, the standard $\Lambda$CDM model best fits the data, the $CC$ model is weakly disfavored, and the $LESC$ model is moderately disfavored.
}

\subsection{Cosmological quantities}

Having determined the optimal parameter values for the $\Lambda$CDM model as well as the LESC and CC models, we proceed to compare their cosmological predictions. For a detailed analysis, outside of the Hubble function and the deceleration parameter, we will also compare the jerk and snap parameters of these models.
\paragraph{\textbf{Hubble function.}} The Hubble function $H(z)$ of the three models is compared with the observational data obtained from \cite{Bouali_2023}.

\begin{figure*}[htbp]
\centering
\includegraphics[width=0.490\linewidth]{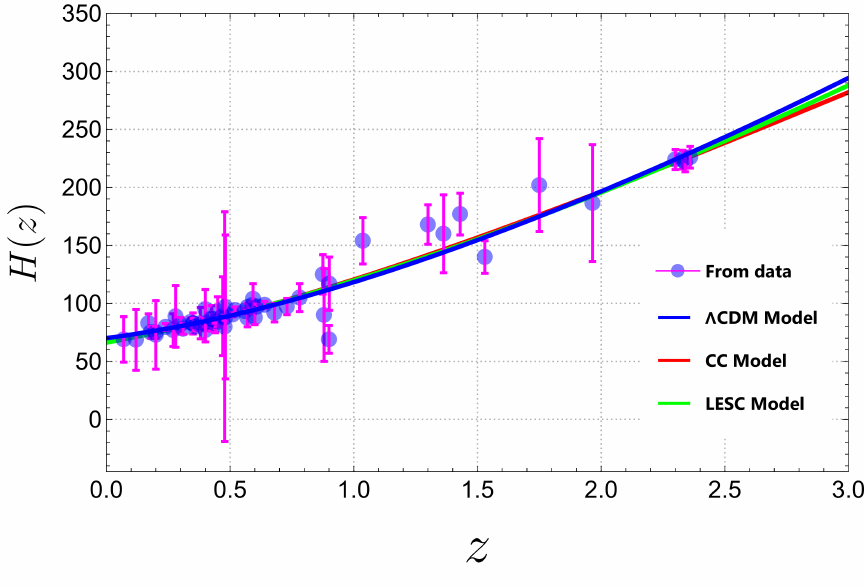} %
\includegraphics[width=0.490\linewidth]{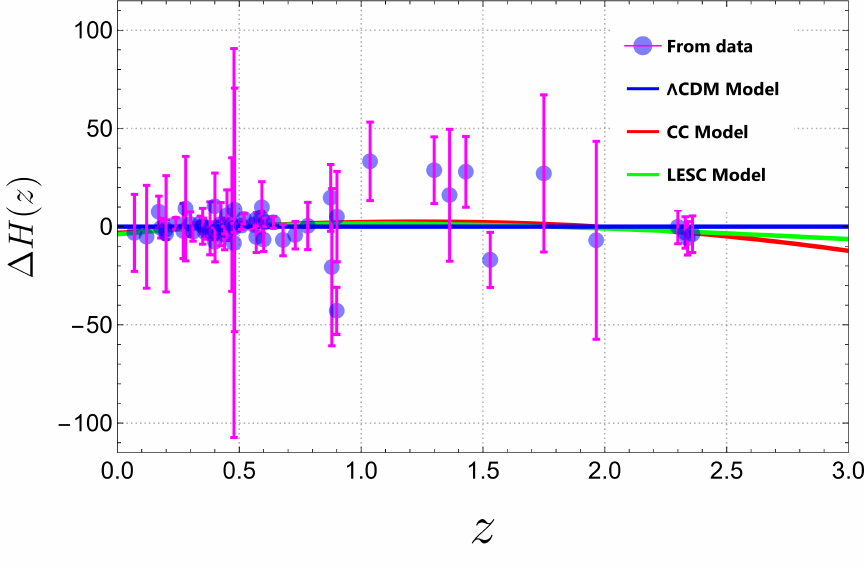}
\caption{Hubble function $H(z)$ for the optimal parameters as a function the redshift $z$ for the $\Lambda$CDM, LESC and CC models (left panel) and the difference plot $H(z)-H_{LCDM}(z)$ (right panel), respectively.}
\label{Hubble}
\end{figure*}

Figure \ref{Hubble}  shows that for the observable range encompassing $z \in (0.07, 2.36)$ the match between the three models and the observational data is perfect, while for larger $z$ values the difference becomes observable.

\paragraph{{\textbf{Deceleration parameter}}} The deceleration parameter $q(z)$ is a dimensionless measure of the rate at which the expansion of the Universe is slowing down ($q>0$), or speeding up ($q<0$). It is defined in terms of the second derivative of the scale factor $a(t)$ with respect to time
\begin{equation}
    q=-\frac{\ddot aa}{a^2}=-\frac{\dot H}{H^2}-1=(1+z)h(z) \frac{ d h(z)}{dz}\frac{1}{h(z)^2} -1.
\end{equation}

The redshift dependence of $q(z)$ for the three distinct models is illustrated in Fig.~\ref{decelerationmatterplot}. Compared to the $\Lambda$CDM model, both the LESC and CC models predict a slightly larger value of the deceleration parameter for small redshift values $0<z<1$. For higher redshifts $z>1$, the LESC and CC models predict a slightly lower, but positive value of $q(z)$.

\begin{figure*}[htbp]
\centering
\includegraphics[width=0.490\linewidth]{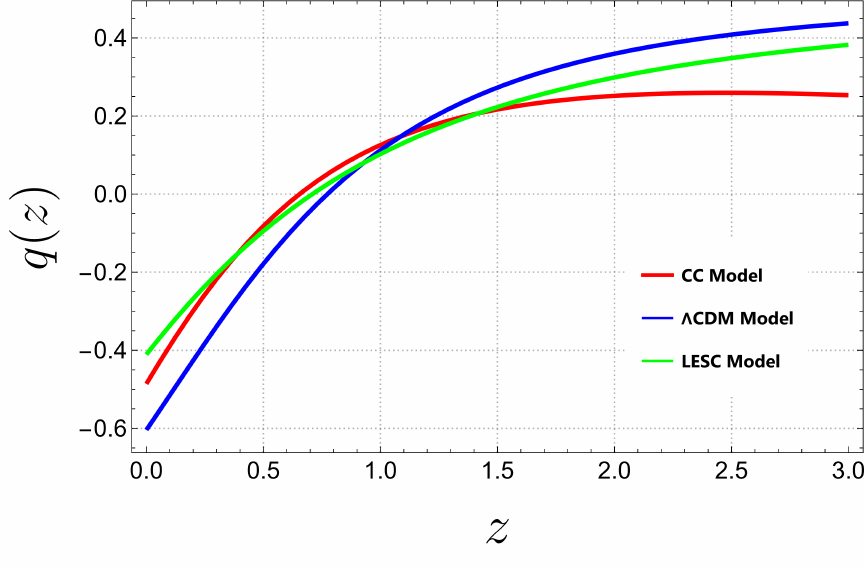} %
\includegraphics[width=0.490\linewidth]{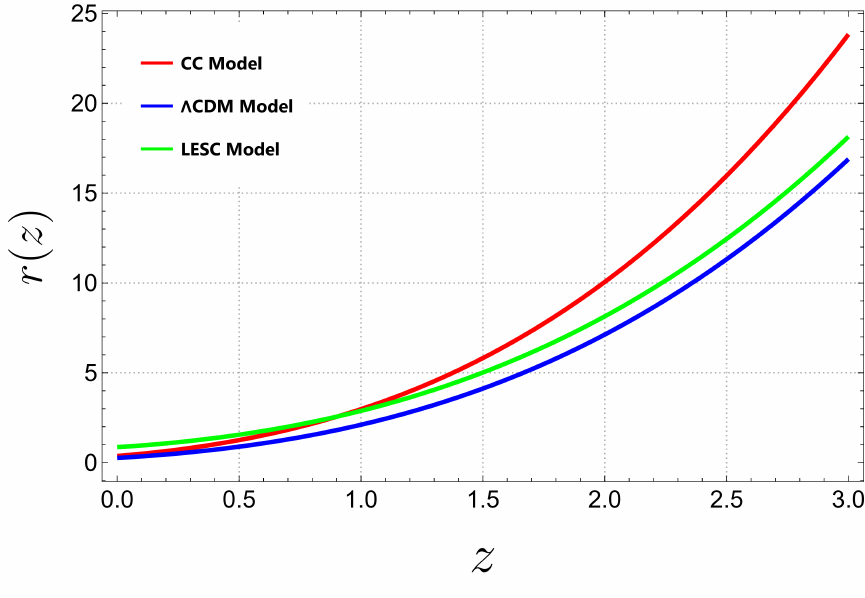}

\caption{The deceleration parameter $q(z)$ as a function of the redshft $z$ (left panel), and the matter density $r(z)$ (right panel) for the $\Lambda$CDM, LESC and CC models with optimal parameters.}
\label{decelerationmatterplot}
\end{figure*}

\paragraph{\textbf{Matter density.}} The matter densities of the $\Lambda$CDM, LESC and CC model are depicted in Fig.~\ref{decelerationmatterplot}. The predictions of the models basically coincide at low redshifts, up to $z \simeq 0.5$, however the LESC and CC cosmological models predict a larger value of matter density at higher redshifts. The LESC model's predictions at high redshifts are closer to the predictions of the $\Lambda$CDM than the ones of the CC model.

\paragraph{\textbf{Jerk.}} The jerk and snap parameters provide insights beyond the deceleration parameter, considering also third and fourth order derivatives of the scale factor \cite{SF1,SF2}. Formally, the jerk parameter is defined as a higher order derivative of the scale factor
\begin{equation}
    j=\frac{1}{a}\frac{d^3a}{d \tau} \left[ \frac{d a}{d \tau}\right]^{-3}=q(2q+1)+(1+z) \frac{dq}{dz}.
\end{equation}

\begin{figure*}[htbp]
\centering
\includegraphics[width=0.490\linewidth]{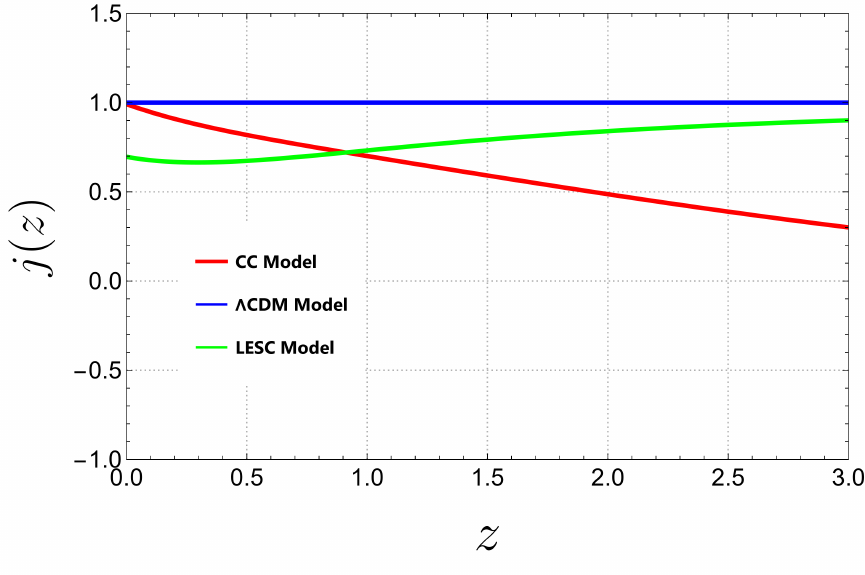} %
\includegraphics[width=0.490\linewidth]{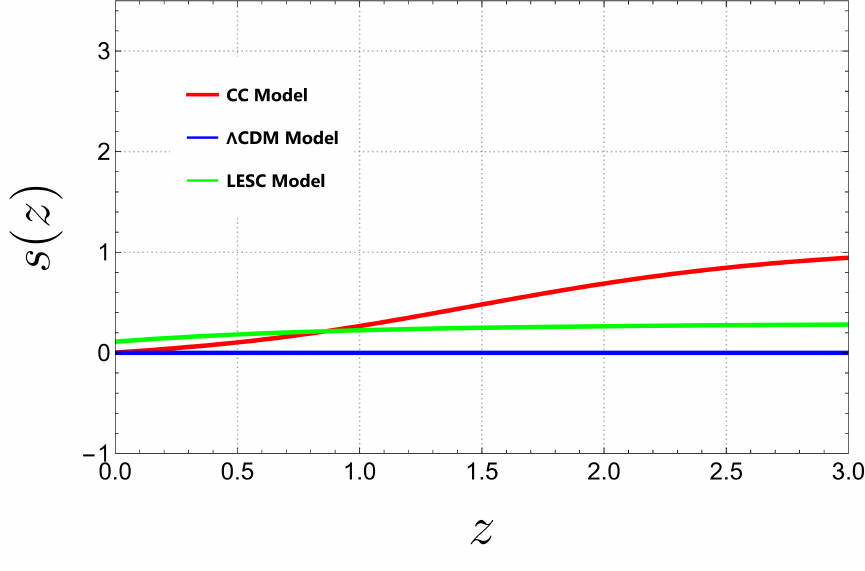}
\caption{The jerk $j(z)$ (left panel) and snap $s(z)$ parameters as functions of the redshift $z$ of the $\Lambda$CDM,LESC and CC models with optimal parameters.}
\label{jerksnapplot}
\end{figure*}

\paragraph{\textbf{Snap}} The snap parameter is a higher-order dimensionless quantity that measures the rate of change of the jerk parameter. Mathematically, it is defined as
\begin{equation}
s = \frac{1}{a} \frac{d^4 a}{d\tau^4} \left[ \frac{1}{a} \frac{da}{d\tau} \right]^{-4} = \frac{j - 1}{3 \left( q - \frac{1}{2} \right)}.
\end{equation}

The jerk and snap parameters of the $\Lambda$CDM, LESC and CC models can be seen on Fig.~\ref{jerksnapplot}. There are significant differences between the cosmological behaviours of these parameters, indicating the possibility of the detailed testing of these cosmological models once high quality observational cosmological data will be available.

\paragraph{{\textbf{$\mathbf{Om(z)}$ diagnostic}}} The $Om(z)$ diagnostic function introduced by Sahni et. al. \cite{sahni2008two} is an important tool in differentiating between alternative cosmological models, and in the comparison of the models with $\Lambda$CDM. The $Om(z)$ function is given by
\begin{equation}
Om(z) = \frac{H^2(z)/H_0^2 - 1}{(1+z)^3 - 1} = \frac{h^2(z) - 1}{(1+z)^3 - 1}.
\end{equation}

For the $\Lambda$CDM model, the $Om(z)$ function is a constant. A positive slope indicates a phantom-like evolution of the dark energy, while a negative slope indicates  quintessence-like dynamics.

\begin{figure*}[htbp]
\centering
\includegraphics[width=0.5\linewidth]{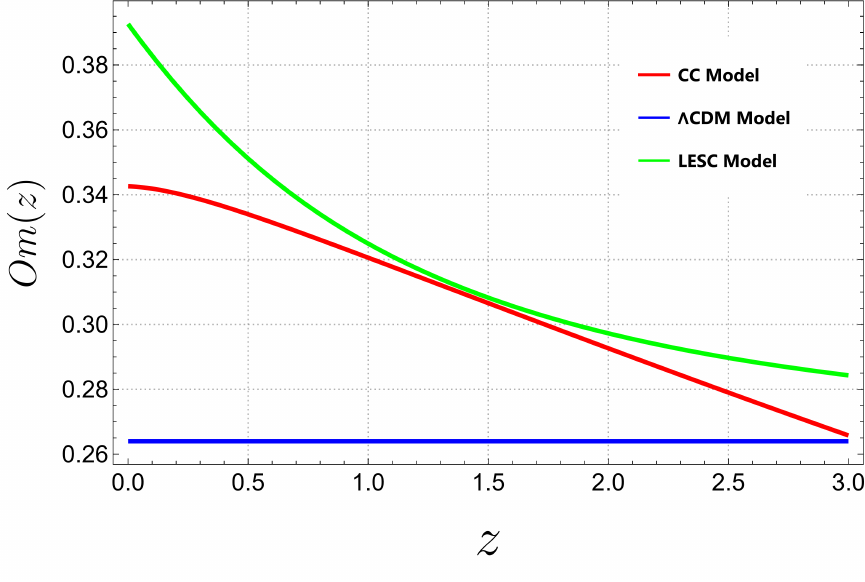}
\caption{The $Om(z)$ diagnostic as a function of the redshift $z$ for the optimal parameter values of the $\Lambda$CDM, LESC and CC models, respectively.}
\label{OmzPlot}
\end{figure*}

The $Om(z)$ diagnostic functions of the LESC and CC models are shown in Fig.~\ref{OmzPlot}. For both of these models, $Om(z)$ has a negative slope throughout the cosmological evolution, which indicates a quintessence-like behaviour.

\paragraph{\textbf{Torsion and related vectors}} The torsion vector $\Psi$ and its derivative $u$ are represented in Fig.~\ref{PsiPlotuPlot}. For the conservative model, $\Psi$ is decreasing monotonically  up to the redshift  $z\approx 0.5$, and becomes an increasing function afterwards. For the linear equation of state model, the torsion  is monotonically decreasing. In either case, up to the redshift $z \simeq 3$, they take positive values. For the CC model, $u$ is also monotonically decreasing, but takes negative values during the cosmological evolution.

\paragraph{\textbf{Statefinder pairs.}} The Statefinder pairs $(j,s)$ and $(j,q)$ are represented in Fig.~\ref{StatefinderPlots}.  The evolution of both pairs
show a significant difference between the two biconnection cosmological models, and $\Lambda$CDM. While as a function of $s$, $j$ is a decreasing function for the CC model, for the LESC model $j$ increases, after a very short decreasing phase. A similar behaviour can be observed for the dependence of $j$ on $q$. Both pairs have at the present time numerical values relatively closed to the $\Lambda$CDM point. The $(j,s)$ plot indicates that both the CC and the LESC models are in the quintessence region.

\begin{figure*}[htbp]
\centering
\includegraphics[width=0.490\linewidth]{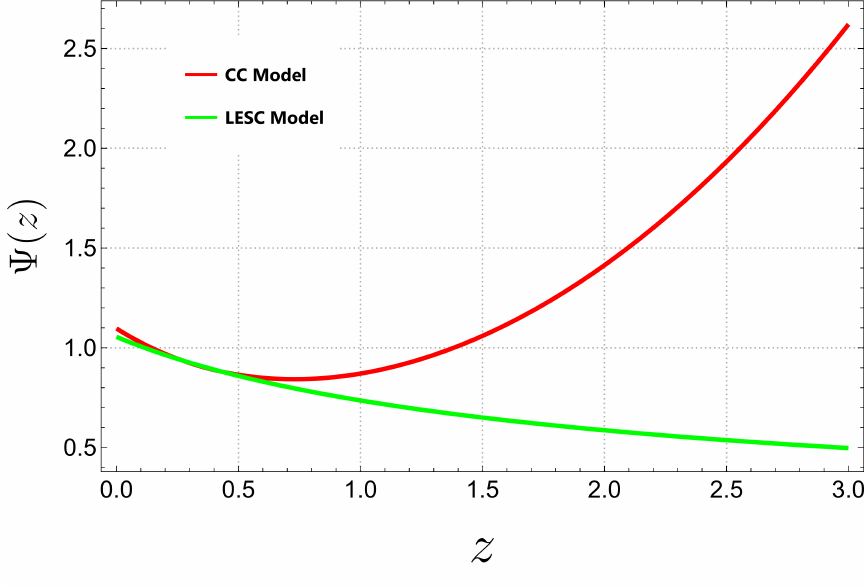} %
\includegraphics[width=0.490\linewidth]{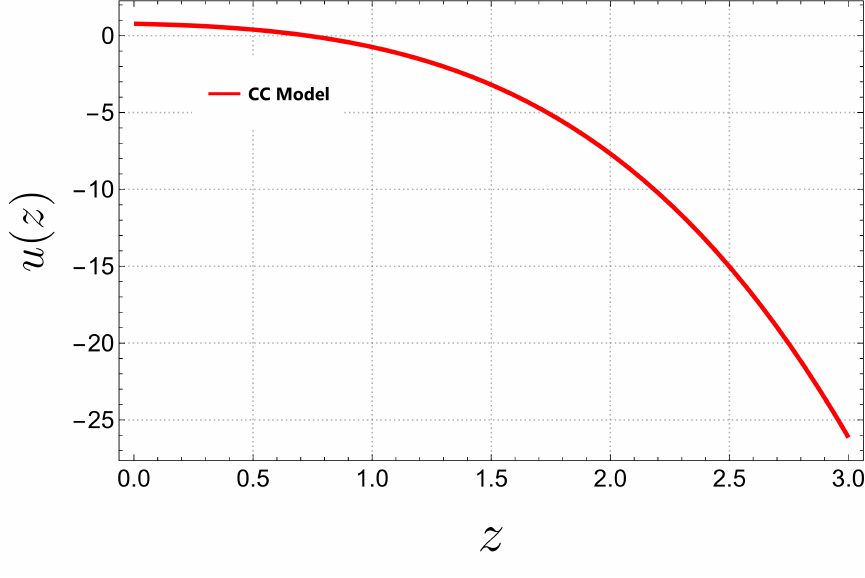}
\caption{The evolution of the torsion vector $\Psi(z)$ (left panel) and of $u(z)$ (right panel) as functions of the redshift variable $z$ for the LESC and CC models with optimal parameters.}
\label{PsiPlotuPlot}
\end{figure*}

\begin{figure*}[htbp]
\centering
\includegraphics[width=0.490\linewidth]{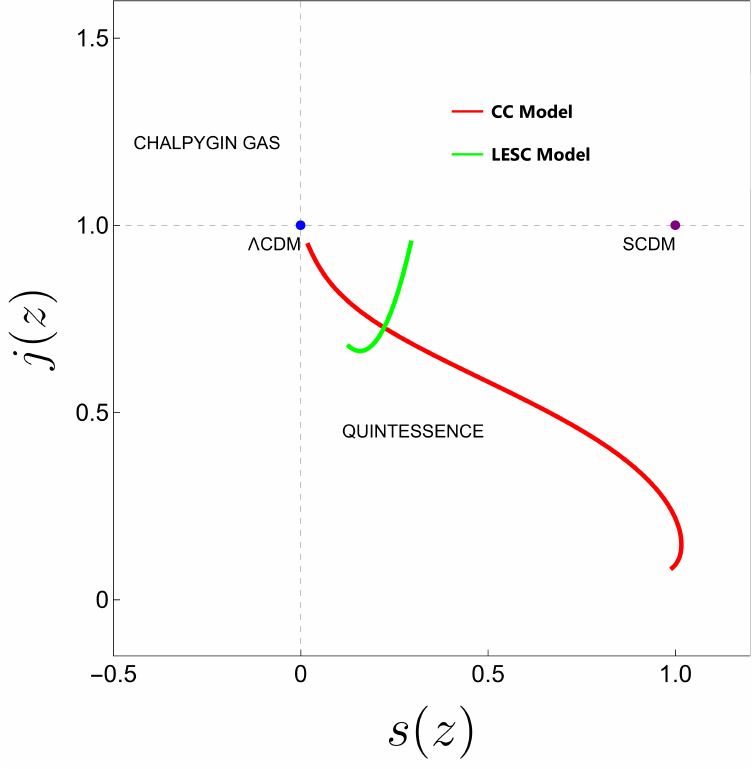} %
\includegraphics[width=0.490\linewidth]{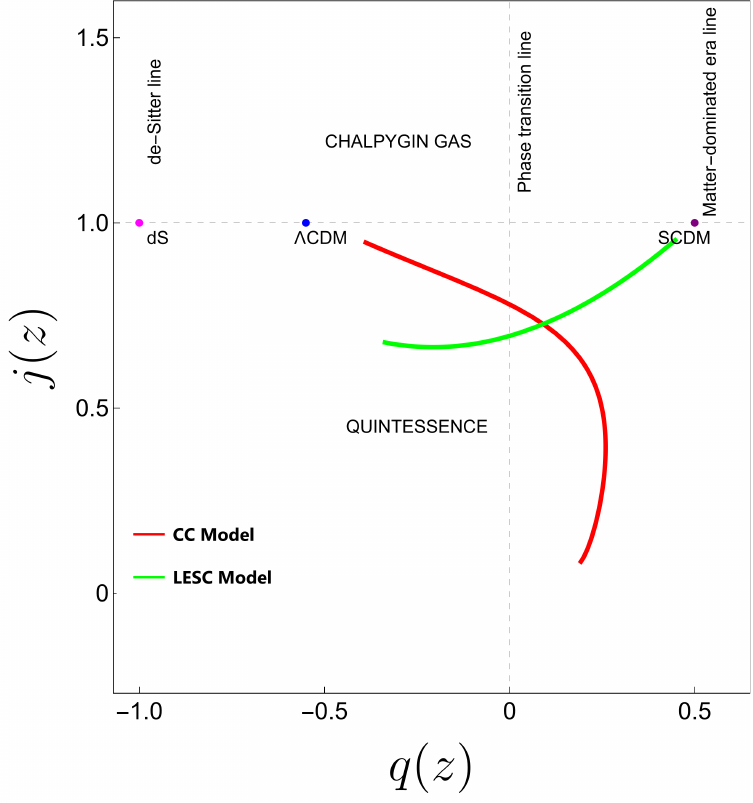}
\caption{The jerk parameter $j(z)$ as a function of the snap  $s(z)$(left panel) and the jerk parameter $j(z)$ as a function of the deceleration parameter $q(z)$ for the LESC and CC models with optimal parameters.}
\label{StatefinderPlots}
\end{figure*}

\paragraph{\textbf{Age of the Universe}.}
The age of the Universe is an important prediction of a cosmological model. It can be directly computed from the Hubble function as
\begin{equation}\label{ageoftheuniverse}
    t_U=\frac{1}{H_0} \lim_{z \to \infty} \int_{0}^{z} \limits \frac{dz'}{(1+z')H(z')}.
\end{equation}

For certain models where $H(z)$ takes an analytic form, $t_U$ can be obtained by a direct computation, without recurring to numerical methods. However, in our case, since $H(z)$ is obtained as the solution of a differential equation, we must calculate the age of the universe numerically. We obtain the following results:
\begin{enumerate}
    \item LESC model: $t_U \simeq 1.41 \times 10^{10}$ years,
    \item CC model: $t_U \simeq 1.45 \times 10^{10}$ years,
    \item $\Lambda$CDM model: $t_U \simeq 1.39 \times 10^{10}$ years.
\end{enumerate}

\section{Discussions and final remarks}\label{sect5}
{
In this paper, we investigated the gravitational and cosmological implications of a biconnection model inspired by the recently introduced mutual curvature tensor.   More specifically, we have considered as a starting point the notion of a dual  connection $\nabla^{*}$ associated with a connection $\nabla$ on a specific geometric structure (manifold). These dually coupled connections are well-known and have been studied in the context of statistical manifolds. On a statistical manifold a cubic tensor $C(X,Y,Z)=(\nabla_X g)(Y,Z)$ is defined, which is completely symmetric. Using this tensor, one can construct two connections $\Gamma^{(1,2)}=\overset{\circ}{\Gamma} \pm C$, whose average is the Levi-Civita connection $\overset{\circ}{\Gamma}$. Thus, the field $C$ describes the deviations from standard Riemannian geometry. The connections $\nabla$ and $\nabla^{*}$ described by $\Gamma^{(1,2)}$ are said to be dual to each other.
}

{
In the case of statistical manifolds one usually imposes the vanishing of the torsion tensors $T,T^*$ of the connection and its dual, respectively. In the present study, we have relaxed this condition by considering statistical manifolds, or information geometries, with torsion and nonmetricity. These type of structures are called quasi-statistical manifolds \cite{kurose2007}, and they significantly enlarge the field of information geometry, also opening some new perspectives for physical applications. We have also introduced a specific quasi-statistical manifold structure, which is based on the Schrödinger connection \cite{Schrod, Klemm, Ming2024,csillag2024schrodinger}, a (still) little known extension of Weyl geometry, initially proposed by Schr\"{o}dinger \cite{Schrod},  which conserves, in the presence of nonmetricity, the length of the vectors under autoparallel transport. The Schrödinger connection, as well as its dual, are fully determined by a vector field. In this framework, once the connection coefficients of the Schrödinger connection and its dual were determined, we computed the mutual curvature tensor, and the mutual curvature scalar of these two connections, obtaining all the necessary tools to build a novel gravitational field theory.
}

{In our approach we have postulated the field equations as having the same form as the standard Einstein field equations, but with the Ricci tensor and scalar being replaced by the mutual curvature tensor and mutual curvature scalar, respectively. This leads to a set of field equations that generalize the Einstein's theory by adding some new, torsion dependent terms into the traditional gravitational field equations as formulated in Riemannian geometry.
}

We have investigated in detail the physical properties and implications of the new field equations. First of all, one should point out that in the present theory the matter energy-momentum tensor is not conserved, since the divergence of $T_{\mu \nu}$ is non-zero. This is a situation specific mostly to gravitational theories with geometry-matter coupling \cite{fRLm,fRT}, but it also appears in the present context. {Energy-momentum preserving models can be constructed in the proposed length-preserving biconnection gravity as well, if this condition is added to the basic theory.} The non-conservation of the matter energy-momentum tensor leads to the non-geodesic motion of massive test particles, and to the presence of an extra force, which depends on the torsion vector. We have explicitly obtained the equations of motion of the massive particles, and we have considered the Newtonian limit of the fluid motion equations for dust. In this case the extra-force depends not only on the torsion vector, but also on the fluid density.

The cosmological implications of the theory have been investigated in detail. As a first step in this analysis the generalized Friedmann equations {of the length-preserving biconnection theory} have been obtained, by assuming a FLRW type geometry, and a particular form of the torsion vector, which preserves the homogeneity and isotropy of the spacetime. The extra terms in the generalized Friedmann equations can be considered as describing the energy density and pressure of an effective, geometric type,  dark energy. In order to close the system of cosmological field equations, one must introduce a supplementary condition/relation for the torsion vector. We have considered two specific cosmological models, in which, for the first model,  we have required the conservation of the matter energy-momentum tensor, while in the second model we have assumed the existence of a linear equation of state relating the pressure and energy density of the dark energy. In both cases we have confronted the theoretical predictions of the models with a set of 57 observational data points for the Hubble function, as well as with the results obtained from fitting the $\Lambda$CDM model with the same data. Both models can be considered as giving an acceptable description of the observational data in the redshift range $z\in (0,2.5)$, with the differences increasing for larger redshifts. Differences do also appear in the behaviours of the deceleration parameter, and of the matter density. The {length preserving biconnection  models}  predict a present day value of the Hubble function in the range $66.2-66.95$ km/s/Mpc, values which are quite closed to the value $H(0)=67.4\pm 0.5$ km/s/Mpc, obtained from the Planck data \cite{Pl1,Pl2}. A recent cosmological model independent determination of $H(0)$, by using Fast Radio Bursts obtained the value $67.3\pm 6.6$ km/s/Mpc.   A similar statistical analysis of the same Hubble data performed in the framework of the $\Lambda$CDM model gives for $H(0)$ the significantly higher value of $70.1$ km/s/Mpc. Hence, the results obtained in the framework of the cosmological models considered in the present work may point towards a possible solution of the Hubble tension, and the need to replace the standard $\Lambda$CDM model.

The age of the Universe, as predicted by {length-preserving biconnection  models} is roughly the same as the age predicted by the standard cosmology, with one model (with the linear equation of state for dark energy) predicting a slightly lower age, while the second  (the conservative matter model) predicting a higher age. The age difference of a few hundred millions of years may also represent a solution to the very important problem of the early formation of the supermassive black holes. In a Universe older by five hundred million years, there will be enough time (more than one billion years), to allow the formation of the early supermassive objects detected by JWST \cite{Hain}. Moreover, the number densities of UV-bright galaxies at $z \geq 10$, inferred from the JWST observations are in tension with the predictions of
the majority of the theoretical models previously developed \cite{Fink1, Fink2}. 

{The generalized Friedmann equations of the Weyl-Schr\"{o}dinger theory \cite{Ming2024} are given by
\be\label{eq_FLRW1}
3H^2+2H\omega-\frac{1}{2}\omega ^2+\dot{\omega}=8\pi\rho
\ee
and
\be\label{eq_FLRW2}
-2\dot{H}-3H^2-\frac{7}{3}H\omega-\frac{1}{6}\omega ^2-\frac{2}{3}\dot{\omega}=8\pi p,
\ee 
respectively, which are again qualitatively similar to the generalized Friedmann equations of the present biconnection gravity theory, given by Eqs.~(\ref{57}) and (\ref{58}), respectively. The differences in the numerical coefficients leads to differences in the present day values of the nonmetricity, and of its relation with the matter density and the present day value of the Hubble function. If the presence of the nonmetricity, and of its time derivative,  could be determined independently from some physical measurements, and since the present day values of the matter density and of the Hubble function are known, the closure equations of the two models, involving different sets of constants, would allow to discriminate between the Weyl-Schr\"{o}dinger and biconnection theories.    
}

The non-geodesic motion of the massive test particles in the present  model may have some implications on the understanding of the dynamics of the massive particles gravitating around the galactic centers. In standard general relativity the gravitational effects due to the presence of an arbitrary matter distribution are described by the term  $a_{N}^{\alpha }=\Gamma _{\mu \nu}^{\alpha }u^{\mu }u^{\nu }$ of the geodesic equation of motion. However, in the Newtonian limit, and in three dimensions, the equation of motion of the massive particles in the {length-preserving biconnection model} is given by  equation (\ref{eq418}), $\vec{a}=\vec{a}_{N}+\vec{a}_E$, where $\vec{a}$ is the total acceleration of the particle, $\vec{a}_{N}$ is the Newtonian gravitational acceleration, and $\vec{a}_E$ is the acceleration due to the presence of the torsion and nonconservation  effects. {Since by definition the extra acceleration is radially oriented, as the gradient of a scalar in spherical symmetry, the extra-force is parallel with the Newtonian gravitational acceleration.} For $\vec{a}_E=0$, the equation of motion reduces  to the Newtonian one, with $\vec{a}=\vec{a}_{N}=-GM\vec{r}/r^{3}$. From the generalized equation of motion in the presence of the extra force we obtain $\vec{a}_E\cdot \vec{a}_{N}=\frac{1}{2}\left(a^{2}-a_{N}^{2}-a_E^{2}\right)$, where the dot denotes the three-dimensional scalar product. Hence, we have obtained the unknown vector $\vec{a}_{N}$ as a function of the total acceleration $\vec{a}$, of the extra acceleration $\vec{a}_E$, and of the magnitudes $a^{2}$, $a_{N}^{2}$ and $a_E^{2}$. One can now obtain the vector $\vec{a }_{N}$ as
\begin{equation}
\vec{a}_{N}=\frac{1}{2}\left( a^{2}-a_{N}^{2}-a_E^{2}\right) \frac{\vec{a}}{aa_E}.
\end{equation}

By assuming that the Newtonian gravitational acceleration is small, $a_{N}\ll a$, we obtain the relation
\begin{equation}
\vec{a}_{N}\approx \frac{1}{2}a\left( 1-\frac{a_E^{2}}{a^{2}}\right) \frac{1}{a_E}\vec{a}.
\end{equation}
By denoting 
\begin{equation}
\frac{1}{a_{M}} \equiv \frac{1}{2a_E}\left( 1-\frac{a_E^{2}}{a^{2}}\right),
\end{equation}
we obtain
\begin{equation}\label{117}
\vec{a}_{N}\approx \frac{a}{a_{M}}\vec{a}.
\end{equation}

Eq. ~(\ref{117}) is similar to the acceleration equation introduced from phenomenological considerations in the  MOND theory \cite{Milgrom}.{ By assuming that $a_E^2/a^2<<1$, we obtain $a_M\approx 2a_E$, and, with the use of Eq.~(\ref{accf}), we find for the Newtonian acceleration the expression
\be
\vec{a}_{N}\approx \frac{a}{2a_E}\vec{a}=\frac{4\pi \rho}{\left|A\right|}a\vec{a}.
\ee
Hence we can write the Newtonian acceleration law in its MOND form as
\be
F_N=ma_N=m\mu \left(\frac{a}{a_0}\right)a,
\ee
where $a_0=\left|A\right|/4\pi \rho$. }

{We can thus obtain $a\approx \sqrt{a_{M}a_{N}}=\sqrt{2a_E}a_N$, and then, by using the Newtonian expression of the gravitational acceleration $
a_{N}=GM/r^{2}$, we find $a\approx \sqrt{2a_{E}GM}/r=v_{tg}^{2}/r$, where by $ v_{tg}$ we have denoted the tangential velocity of the particle. Hence, 
$v_{tg}^{2}\rightarrow v_{\infty }^{2}=\sqrt{2a_{E}GM}$, and from this relation we obtain the Tully-Fisher law $L \sim v_{\infty }^{4}$ in the form
$v_{\infty }^{4}=2a_{E}GM$, where $L$ is the galactic luminosity, which is proportional to the galactic mass \cite{Milgrom}. At large distances from the galactic center $v_{\infty}\approx {\rm constant}$, and hence we can obtain for the torsion vector, in the Newtonian approximation,  the expression
\be\label{117}
\left|\vec{A}\right|\approx \frac{4\pi \rho v_\infty^4}{GM}.
\ee 
}
{A basic difference between MOND and the present approach is that while in MOND $a_0$ is assumed to be a universal constant $a_0=1.2\times 10^{-10}\;{\rm m/s^2}$, in the present approach $a_0$ is a function of distance, depending on the background density, and on the torsion vector. Thus, $a_0$ varies from galaxy to galaxy, thus pointing towards a solution to the problems MOND has in explaining the large variety of the galactic rotation curves. }

Hence, the study of the galactic rotation curves opens another possibility, via the presence of the extra acceleration, to test the presence of the torsion in the {length-preserving biconnection theory}. In the present approach  $a_{M}$ is not a universal constant, as it is in the standard MOND theory, since it depends on the torsion vector.

{  We  would also like to mention that our analysis based on the Newtonian approximation of the equation of motion is valid for a large class of modified gravity theories, and it can be extended to any physical model that introduces an extra acceleration in the equation of motion. But once a physical model that predicts the form of the extra force can be explicitly obtained, Eq.~(\ref{117}) can be used to test, and discriminate between various modified gravity theories. 
}   

In the present paper we have introduced a geometric theory of gravity, which is based on mathematical concepts adopted from information geometry.  The field equations we have considered are a generalization of the Einstein equations of standard general relativity, and differ from them in the empty space, as well as in the
presence of matter, due to the presence of a torsion vector. Hence, the predictions of the present theory  may lead,  to some important differences, as compared to the predictions of Einstein's general relativity, in several problems of current interest, like the  cosmology of the early and late Universe, black holes and  gravitational collapse or in the generation of gravitational waves. The detailed investigations of these physical, astrophysical and cosmological phenomena may also provide specific effects and signatures that could help in distinguishing and discriminating between the various existing gravitational  theories.
\section*{Acknowledgements}

The work of L.Cs. and M.J. is supported by Collegium Talentum Hungary. L.Cs. Would like to thank Xumin Liang, Damianos Iosifidis for the helpful discussions and for the StarUBB research scholarship.
 
\appendix

\section{Coordinate-free treatment of quasi-statistical manifolds}\label{appendixA}
In this Appendix we provide a geometric, coordinate-free approach to quasi-statistical manifolds and the results used in our manuscript. {Let us explicitly mention that most of the results are already present in the literature \cite{Blaga2022}. We present them for completeness, as we have to adapt them to our convention in coordinates, and to make the paper readable and self contained for the physics community.}
\begin{definition}
 Let $(M,g,\nabla)$ be a pseudo-Riemannian manifold. Two affine connections $\nabla$ and $\nabla^{*}$ are said to be \textbf{dual with respect to the metric} if
 \begin{equation}\label{dualconnection}
     X(g(Y,Z))=g\left( \nabla_X Y,Z \right)+ g\left( Y,\nabla_{X}^{*} Z \right)
 \end{equation}
 is satisfied for all vector fields $X,Y,Z$.
\end{definition}
\begin{remark}
    In a local chart with $X=\partial_\mu,Y=\partial_\nu,Z=\partial_\rho$ condition \eqref{dualconnection} takes the form
    \begin{equation}
        \partial_\mu g_{\nu \rho}=\Gamma_{\rho \nu \mu} + \Gamma^{*}_{\nu \rho \mu}.
    \end{equation}
\end{remark}

\begin{definition}
    A pseudo-Riemannian manifold $(M,g)$ equipped with a torsionful affine connection $\nabla$ is called a \textbf{quasi-statistical manifold} if
    \begin{equation}\label{thiscondition}
        \left( \nabla_X g \right)(Y,Z) - \left( \nabla_Y g \right)(X,Z) + g \left(T(X,Y),Z \right)=0
    \end{equation}
    for all vector fields $X,Y,Z$.
\end{definition}
\begin{remark}
    In case $\nabla$ is torsion-free, we recover
    \begin{equation}
        \left( \nabla_X g \right)(Y,Z)=\left( \nabla_Y g \right) (X,Z),
    \end{equation}
    which is the condition for the pair $(M,g,\nabla)$ to be a statistical manifold.
\end{remark}
\begin{remark}
    In a local chart with $X=\partial_\mu, Y=\partial_\nu,Z=\partial_\rho$ condition \eqref{thiscondition} is equivalent to
    \begin{equation}
        \nabla_{\mu} g_{\nu \rho} - \nabla_{\nu} g_{\mu \rho} + T_{\rho \mu \nu}=0.
    \end{equation}
\end{remark}
We now present a theorem, which relates the torsion and nonmetricity of a connection and its dual, but first let us formally introduce these two tensors in a geometric setting.
\begin{definition}
    The torsion of an affine connection $\nabla$ is the vector-valued tensor field defined by
    \begin{equation}
        T(X,Y)=\nabla_{X} Y -\nabla_{Y} X - [X,Y].
    \end{equation}
\end{definition}
\begin{definition}
    The nonmetricity of an affine connection $\nabla$ is the $(0,3)$-tensor field defined by
    \begin{equation}
        Q(X,Y,Z):=\left(-\nabla_X g \right)(Y,Z).
    \end{equation}
\end{definition}
{Although the upcoming theorem directly follows from the definitions of quasi-statistical manifolds, to the best of our knowledge, a detailed proof of it is not present in the literature. A statement of it, using different notations is present in \cite{Blaga2022}.}
\begin{theorem}\label{theoremdual}
    Let $(M,g)$ be a pseudo-Riemannian manifold equipped with two affine connections $(\nabla,\nabla^*)$, which are dual with respect to $g$. Then torsion $T^{*}$ and the nonmetricity $Q^{*}$ of the dual connection $\nabla^{*}$ satisfy the conditions:
    \begin{enumerate}
        \item[$(i)$] nonmetricity of the dual connection: \begin{equation}
        Q^{*}(X,Y,Z)=-Q(X,Y,Z).
        \end{equation}
        \item[$(ii)$] Torsion of the dual connection:
        \begin{equation}
        \begin{aligned}
        g \left(T^{*}(X,Y),Z \right)&=g\left( T(X,Y),Z \right)\\&- Q(X,Y,Z)+ Q(Y,X,Z).
        \end{aligned}
        \end{equation}
    \end{enumerate}
\end{theorem}
\begin{proof}
To prove the first statement, we follow the definitions, and employ the Leibnitz rule two times:
\begin{multline}
      Q^*(X,Y,Z)=\left( - \nabla^*_X g \right)(Y,Z)\\
      =-X(g,(Y,Z))+ g \left( \nabla^*_X Y,Z \right) + g\left(Y,\nabla^*_X Z \right)\\
      =-g(\nabla_X Z,Y)+X(g(Y,Z))-g(\nabla_X Y,Z)\\
      =(\nabla_X g)(Y,Z)=-Q(X,Y,Z).
\end{multline}

   For the second part, as $\nabla$ is dual to $\nabla^{*}$, we have
   \begin{equation}\label{dual1}
       g\left( Y, \nabla^{*}_{X} Z \right)=X(g(Y,Z))-g\left( \nabla_X Y,Z \right),
   \end{equation}
   \begin{equation}\label{dual2}
g\left(Z,\nabla^{*}_{X}Y \right)=X(g(Z,Y))- g \left( \nabla_X Z,Y \right),
   \end{equation}
   \begin{equation}\label{dual3}
       g\left(Z,\nabla^*_Y X \right)=Y(g(Z,X))-\left(\nabla_Y Z,X \right).
   \end{equation}

   Thus, by employing the definition of torsion, we obtain
   \begin{equation}
       g\left(T^{*}(X,Y),Z \right)=g \left( \nabla^*_X Y - \nabla^*_Y X -[X,Y],Z\right).
   \end{equation}

   Expanding the brackets using multilinearity and equations \eqref{dual1},\eqref{dual2} and \eqref{dual3} yields
\begin{multline}
       g\left(T^{*}(X,Y),Z \right)=X(g(Z,Y))-g(\nabla_X Z,Y)\\
       -Y(g(Z,X))+g(\nabla_Y Z,X)-g([X,Y],Z),
\end{multline}
or equivalently
\begin{multline}
    g\left(T^{*}(X,Y),Z \right)= \nabla_X (g(Z,Y))- g\left( \nabla_X Z,Y \right) \\
    -\nabla_Y(g(Z,X)) + g \left( \nabla_Y Z,X \right) - g([X,Y],Z).
\end{multline}

Applying the Leibnitz rule, we get for the right hand side
\begin{multline}
    \textcolor{-green}{(\nabla_X g)(Y,Z) }\textcolor{red}{+g(\nabla_X Z,Y)} {+g(Z,\nabla_X Y)}
    \textcolor{red}{-g(\nabla_X Z,Y)}\\
   \textcolor{-green}{ - (\nabla_Y g)(Z,X) }\textcolor{blue}{- g(\nabla_Y Z,X)} {- g(Z,\nabla_Y X) }\textcolor{blue}{+ g(\nabla_Y Z,X)}\\
   {-g([X,Y],Z)}.
\end{multline}

We identify the orange terms with nonmetricity, by definition. The blue and red terms cancel, leaving us with
\begin{multline}
    g(T^*(X,Y),Z)=\textcolor{-green}{-Q(X,Y,Z)}+g(Z,\nabla_X Y)\\
\textcolor{-green}{+Q(Y,Z,X)} - g(Z,\nabla_Y X) - g([X,Y],Z).
\end{multline}

Using the definition of torsion
\begin{multline}
    g(T^*(X,Y),Z)=g(T(X,Y),Z)\\
    -Q(X,Y,Z)+Q(Y,Z,X)
\end{multline}
we obtain the desired result.

\end{proof}
\begin{corollary}
    In a local chart given by $X=\partial_\mu, Y=\partial_\nu, Z=\partial_\rho$, the nonmetricity and torsion of the dual affine connection satisfy
    \begin{equation}\label{dualcoeffs}
        Q^{*}_{\mu \nu \rho}=- Q_{\mu \nu \rho}, \; \; T^{*}_{\rho \mu \nu}=T_{\rho \mu \nu} -Q_{\mu \nu \rho} + Q_{\nu \mu \rho}.
    \end{equation}
\end{corollary}
\begin{corollary}\label{corollarytorsionalmanifolds}
    Let $(M,g)$ be a pseudo-Riemannian manifold and $\left(\nabla,\nabla^{*} \right)$ be dual with respect to $g$. Moreover, denote with $T$ the torsion tensor of $\nabla$ and with $T^{*}$ the torsion tensor of $\nabla^{*}$. Then, the following statements hold:
    \begin{enumerate}
        \item[$(i)$] $T=0$ iff the pair $(M,g,\nabla^{*})$ is a quasi-statistical manifold.
        \item[$(ii)$]$T^{*}=0$ iff the pair $(M,g,\nabla)$ is a quasi-statistical manifold.
    \end{enumerate}
\end{corollary}
{This immediate consequence is also presented in \cite{Blaga2022}.}

{Let us now prove the well known fact, that for a statistical manifold, the average connection $\overline{\nabla}=\nabla+\nabla^{*}$ is the Levi-Civita connection $\overset{\circ}{\nabla}$.
\begin{proposition}
    Let $\nabla$ and $\nabla^{*}$ be dual connections on a pseudo-Riemannian manifold $(M,g)$. Then, the average connection
    \begin{equation}
        \overline{\nabla}:=\frac{1}{2}\left( \nabla+\nabla^* \right))
    \end{equation}
    is $g$-metrical, that is, $(\overline{\nabla}_{X}g)(Y,Z)=0$ for all vector fields $X,Y,Z$.
\end{proposition}
\begin{proof}
    By the Leibnitz rule, we have
    \begin{equation}
        \left(\overline{\nabla}_{X} g \right)(Y,Z)=X(g(Y,Z))-g \left( \overline{\nabla}_{X} Y,Z \right) - g \left( Y, \overline{\nabla}_{X} Z \right).
    \end{equation}
    Using the definition of the average connection $\overline{\nabla}$
\begin{equation}
\begin{aligned}
        \left(\overline{\nabla}_{X} g \right)(Y,Z)=&X(g(Y,Z))\\&-\frac{1}{2} \left(g(\nabla_X Y,Z)+g\left( \nabla^{*}_{X} Y,Z\right) \right)\\
        &-\frac{1}{2} \left(g(Y,\nabla_X Z)+g \left( Y,\nabla^{*}_{X} Z \right) \right).
\end{aligned}
\end{equation}
Regrouping the terms on RHS yields
\begin{equation}
\begin{aligned}
      &\frac{1}{2} \left(X(g(Y,Z))-g\left( \nabla_X Y,Z \right) - g\left( Y, \nabla^*_X Z \right) \right)\\
      +&\frac{1}{2}\left( X(g(Y,Z))-g \left( \nabla^*_X Y,Z  \right) - g(Y,\nabla_X Z) \right),
\end{aligned}
\end{equation}
which is identically zero, since $\nabla$ and $\nabla^{*}$ are dual connections. We thus have shown
\begin{equation}
    \left(\overline{\nabla}_X g \right)(Y,Z)=0+0=0.
\end{equation}
\end{proof}
\begin{proposition}
    If $(M,g,\nabla)$ is a statistical manifold, i.e. $\nabla$ is torsion-free and the cubic tensor $C(X,Y,Z)=(\nabla_X g)(Y,Z)$ is totally symmetric, then the dual connection $\nabla^{*}$ is torsion-free.
\end{proposition}
\begin{proof}
    By the the Leibnitz rule and the definition of cubic tensor, we immediately obtain
    \begin{equation}
        C(Z,X,Y)=Z(g(X,Y))-g(\nabla_Z X,Y) - g(X,\nabla_Z Y),
    \end{equation}
    \begin{equation}
        C(Y,Z,X)=Y(g(Z,X))-g(\nabla_Y Z,X) - g(Z,\nabla_Y X).
    \end{equation}
    Substracting  these two equations leads to
    \begin{equation}
    \begin{aligned}
        &C(Z,X,Y)-C(Y,Z,X)=Z(g(X,Y))-g(\nabla_Z X,Y)\\
        &- \underbrace{g(X,\nabla_Z Y)}_{1} - Y(g(Z,X))+\underbrace{g(\nabla_Y Z,X)}_{2}+g(Z,\nabla_Y X).
    \end{aligned}
    \end{equation}
    Using that $\nabla$ is torsion-free, terms $1$ and $2$ can be combined to give
    \begin{equation}
    \begin{aligned}
        &C(Z,X,Y)-C(Y,Z,X)=Z(g(X,Y))-Y(g(Z,X)) \\
        &-g(\nabla_Z X,Y) + g(Z,\nabla_Y X)-g(X,[Z,Y])..
    \end{aligned}
    \end{equation}
    Using the definition of torsion of $T^{*}$ yields
    \begin{equation}
        \begin{aligned}
           & C(X,Y,Z)-C(Y,Z,X)=\underbrace{Z(g(X,Y))}_{1}- \underbrace{Y(g(Z,X))}_{2}\\
           &\underbrace{-g(\nabla_Z X,Y)}_{1}+\underbrace{g(Z,\nabla_Y X)}_{2}-\underbrace{g(X,\nabla^*_Z Y)}_{1}\\
           &+\underbrace{g(X,\nabla^*_Y Z)}_{2} +g\left( X,T^{*}(Z,Y) \right).
        \end{aligned}
    \end{equation}
    Terms $1,2$ vanish by the definition of dual connections. We are thus left with
    \begin{equation}
        C(X,Y,Z)-C(Y,Z,X)=g(X,T^{*}(Z,Y)).
    \end{equation}
    Since $C$ is completely symmetric, we immediately get
    \begin{equation}
        g(X,T^{*}(Z,Y))=0,
    \end{equation}
    and as $g$ is non-degenerate, this leads to the desired result
    \begin{equation}
        T^*=0.
    \end{equation}
\end{proof}
}

\section{Computation of Ricci- and mutual difference tensors }\label{appendixB}
Recall that the Schrödinger connection is torsion-free and has nonmetricity given by
\begin{equation}
    Q_{\mu \nu \rho}=\pi_{\mu} g_{\nu \rho} - \frac{1}{2} \left( g_{\mu \nu} \pi_{\rho} + g_{\mu \rho} \pi_{\nu} \right).
\end{equation}

Relabeling indices, we immediately obtain
\begin{equation}
\begin{aligned}
    -Q_{\lambda \nu \rho}&=\textcolor{red}{-\pi_{\lambda} g_{\nu \rho}} \textcolor{blue}{+ \frac{1}{2} g_{\lambda \nu} \pi_\rho} \textcolor{orange}{ + \frac{1}{2}g_{\lambda \rho} \pi_{\nu}},\\
    Q_{\rho \lambda \nu}&=\textcolor{blue}{\pi_{\rho} g_{\lambda \nu}} \textcolor{orange}{ -\frac{1}{2} g_{\rho \lambda} \pi_{\nu} }\textcolor{red}{- \frac{1}{2}g_{\rho \nu} \pi_{\lambda}},\\
    Q_{\nu \rho \lambda}&= \textcolor{orange}{\pi_{\nu} g_{\rho \lambda} } \textcolor{red}{-\frac{1}{2} g_{\nu \rho} \pi_{\lambda}}  \textcolor{blue}{-\frac{1}{2} g_{\nu \lambda} \pi_{\rho}}.
\end{aligned}
\end{equation}

Hence, for the distortion tensor of Schrödinger geometry we have
\begin{equation}
    \tensor{N}{^\mu _\nu _\rho}=\frac{1}{2} g^{\lambda \mu} \left( \textcolor{red}{-2 \pi_\lambda g_{\nu \rho}}
    \textcolor{blue}{+ \pi_\rho g_{\lambda \nu}}  \textcolor{orange}{+ \pi_{\nu} g_{\rho \lambda}}\right),
\end{equation}
which can be equivalently rewritten as
\begin{equation}\label{distortion}
    \tensor{N}{^\mu_\nu _\rho}=-\pi^{\mu} g_{\nu \rho} + \frac{1}{2} \pi_\rho \delta^\mu_\nu + \frac{1}{2} \pi_\nu \delta^{\mu}_{\rho}.
\end{equation}

Following the same steps, one can show that the distortion tensor of the dual is
\begin{equation}\label{distortiondual}
\begin{aligned}
    \tensor{N}{^{*}^{\mu}_{\nu}_\rho}
    =-\frac{1}{2} \pi^{\mu} g_{\nu \rho} - \frac{1}{2} \pi_{\rho} \delta^{\mu}_{\nu} +  \pi_{\nu} \delta^{\mu}_{\rho}.
\end{aligned}
\end{equation}

Substracting  \eqref{distortiondual} from  \eqref{distortion} yields
\begin{equation}
    \tensor{K}{^\mu _\nu _\rho}=-\frac{1}{2} \pi^\mu g_{\nu \rho}+ \pi_{\rho} \delta^{\mu}_{\nu} -\frac{1}{2} \pi_{\nu} \delta^{\mu}_{\rho}.
\end{equation}

Hence, the mutual difference tensor is given by
\begin{equation}
    \tensor{D}{_\rho _\nu}=-\frac{1}{4} \pi_{\nu} \pi_{\rho} - \frac{1}{2} \pi_{\sigma} \pi^{\sigma} g_{\rho \nu},
\end{equation}
while the mutual difference scalar reads
\begin{equation}
    D=-\frac{9}{4} \pi_{\sigma} \pi^{\sigma}.
\end{equation}

We now move on to compute the Ricci scalar of the Schrödinger connection, and its dual, respectively. In the presence of distortion, the Ricci tensor of an affine connection can be written as
\begin{equation}\label{riccitens}
\begin{aligned}
      R_{\mu \nu}=\overset{\circ}{R}_{\mu \nu} &+\overset{\circ}{\nabla}_{\alpha} \tensor{N}{^\alpha _\nu _\mu}-\overset{\circ}{\nabla}_{\nu} \tensor{N}{^\alpha _\alpha _\mu}\\
      &+\tensor{N}{^\alpha_\alpha  _\rho} \tensor{N}{^\rho_\nu_\mu} - \tensor{N}{^\alpha _\nu _\rho}\tensor{N}{^\rho_\alpha_\mu}.
\end{aligned}
\end{equation}

We obtain the Ricci tensor of the Schrödinger connection by substituting  the distortion tensor \eqref{distortion} into \eqref{riccitens}
\begin{equation}
\begin{aligned}
R_{\mu \nu}=\overset{\circ}{R}_{\mu \nu} &- g_{\mu \nu} \overset{\circ}{\nabla}_{\alpha} \pi^{\alpha} + \frac{1}{2} \overset{\circ}{\nabla}_{\mu} \pi_{\nu} -\overset{\circ}{\nabla}_{\nu} \pi_{\mu} \\
&  - \frac{1}{2}
 \pi_{\rho} \pi^{\rho} g_{\mu \nu} - \frac{1}{4} \pi_{\mu} \pi_{\nu}.
 \end{aligned}
\end{equation}

It immediately follows that the Ricci scalar is given by
\begin{equation}
    R=\overset{\circ}{R}-\frac{9}{2} \overset{\circ}{\nabla}_{\alpha} \pi^{\alpha} - \frac{9}{4} \pi_{\alpha} \pi^{\alpha}.
\end{equation}

For the dual Schrödinger connection, we have
\begin{equation*}
\begin{aligned}
    R_{\mu \nu}^{\star}
        =\overset{\circ}{R}_{\mu \nu}&+\overset{\circ}{\nabla}_\alpha\left(-\frac{1}{2}\pi^{\alpha}g_{\nu\mu}-\frac{1}{2}\pi_{\mu}\delta_{\nu}^{\alpha}+\pi_{\nu}\delta_{\mu}^{\alpha}\right)\\
        &-\overset{\circ}{\nabla}_\nu \left(-\frac{1}{2}\pi^{\alpha}g_{\alpha\mu}-\frac{1}{2}\pi_{\mu}\delta_{\alpha}^{\alpha}+\pi_{\alpha}\delta_{\mu}^{\alpha}\right)\\
        &+\left(-\frac{1}{2}\pi^{\alpha}g_{\alpha\rho}-\frac{1}{2}\pi_{\rho}\delta_{\alpha}^{\alpha}+\pi_{\alpha}\delta_{\rho}^{\alpha}\right) \times\\
        &\times \left(-\frac{1}{2}\pi^{\rho}g_{\nu\mu}-\frac{1}{2}\pi_{\mu}\delta_{\nu}^{\rho}+\pi_{\nu}\delta_{\mu}^{\rho}\right)\\
        &-\left(-\frac{1}{2}\pi^{\alpha}g_{\nu\rho}-\frac{1}{2}\pi_{\rho}\delta_{\nu}^{\alpha}+\pi_{\nu}\delta_{\rho}^{\alpha}\right)  \times\\
        &\times \left(-\frac{1}{2}\pi^{\rho}g_{\alpha\mu}-\frac{1}{2}\pi_{\mu}\delta_{\alpha}^{\rho}+\pi_{\alpha}\delta_{\mu}^{\rho}\right)\\
        =\overset{\circ}{R}_{\mu \nu}&-\frac{1}{2}\overset{\circ}{\nabla}_\alpha\pi^{\alpha}g_{\nu\mu}-\frac{1}{2}\overset{\circ}{\nabla}_\nu\pi_{\mu}+\overset{\circ}{\nabla}_\mu\pi_{\nu}\\
        &+\frac{1}{2}\overset{\circ}{\nabla}_\nu \pi_{\mu}+2\overset{\circ}{\nabla}_\nu \pi_{\mu}-\overset{\circ}{\nabla}_\nu \pi_{\mu}\\
        &+\frac{1}{2}\pi^{\alpha}g_{\alpha\rho}\frac{1}{2}\pi^{\rho}g_{\nu\mu}+\frac{1}{2}\pi_{\rho}\delta_{\alpha}^{\alpha}\frac{1}{2}\pi^{\rho}g_{\nu\mu}\\
        &-\pi_{\alpha}\delta_{\rho}^{\alpha}\frac{1}{2}\pi^{\rho}g_{\nu\mu}
        +\frac{1}{2}\pi^{\alpha}g_{\alpha\rho}\frac{1}{2}\pi_{\mu}\delta_{\nu}^{\rho}\\
        &+\frac{1}{2}\pi_{\rho}\delta_{\alpha}^{\alpha}\frac{1}{2}\pi_{\mu}\delta_{\nu}^{\rho}-\pi_{\alpha}\delta_{\rho}^{\alpha}\frac{1}{2}\pi_{\mu}\delta_{\nu}^{\rho}\\
        &-\frac{1}{2}\pi^{\alpha}g_{\alpha\rho}\pi_{\nu}\delta_{\mu}^{\rho}-\frac{1}{2}\pi_{\rho}\delta_{\alpha}^{\alpha}\pi_{\nu}\delta_{\mu}^{\rho}\\
        &+\pi_{\alpha}\delta_{\rho}^{\alpha}\pi_{\nu}\delta_{\mu}^{\rho}
        -\frac{1}{2}\pi^{\alpha}g_{\nu\rho}\frac{1}{2}\pi^{\rho}g_{\alpha\mu}\\
        &-\frac{1}{2}\pi_{\rho}\delta_{\nu}^{\alpha}\frac{1}{2}\pi^{\rho}g_{\alpha\mu}+\pi_{\nu}\delta_{\rho}^{\alpha}\frac{1}{2}\pi^{\rho}g_{\alpha\mu}\\
        &-\frac{1}{2}\pi^{\alpha}g_{\nu\rho}\frac{1}{2}\pi_{\mu}\delta_{\alpha}^{\rho}-\frac{1}{2}\pi_{\rho}\delta_{\nu}^{\alpha}\frac{1}{2}\pi_{\mu}\delta_{\alpha}^{\rho}\\
        &+\pi_{\nu}\delta_{\rho}^{\alpha}\frac{1}{2}\pi_{\mu}\delta_{\alpha}^{\rho}
        +\frac{1}{2}\pi^{\alpha}g_{\nu\rho}\pi_{\alpha}\delta_{\mu}^{\rho}\\
        &+\frac{1}{2}\pi_{\rho}\delta_{\nu}^{\alpha}\pi_{\alpha}\delta_{\mu}^{\rho}-\pi_{\nu}\delta_{\rho}^{\alpha}\pi_{\alpha}\delta_{\mu}^{\rho}\\
        &=\overset{\circ}{R}_{\mu \nu}-\frac{1}{2}g_{\mu\nu}\overset{\circ}{\nabla}_\alpha\pi^{\alpha}+\overset{\circ}{\nabla}_\mu\pi_{\nu}+\overset{\circ}{\nabla}_\nu \pi_{\mu}\\
        &+\frac{1}{4}\pi_{\alpha}\pi^{\alpha}g_{\mu\nu}+\pi_{\rho}\pi^{\rho}g_{\mu\nu}-\frac{1}{2}\pi_{\rho}\pi^{\rho}g_{\mu\nu}\\
        &+\frac{1}{4}\pi_{\nu}\pi_{\mu}+\pi_{\nu}\pi_{\mu}-\frac{1}{2}\pi_{\nu}\pi_{\mu}\\
        &-\frac{1}{2}\pi_{\mu}\pi_{\nu}-2\pi_{\mu}\pi_{\nu}+\pi_{\mu}\pi_{\nu}\\
        &-\frac{1}{4}\pi_{\mu}\pi_{\nu}-\frac{1}{4}\pi_{\rho}\pi^{\rho}g_{\nu\mu}+\frac{1}{2}\pi_{\nu}\pi_{\mu}\\
        &-\frac{1}{4}\pi_{\nu}\pi_{\mu}-\frac{1}{4}\pi_{\nu}\pi_{\mu}+2\pi_{\nu}\pi_{\mu}\\
        &+\frac{1}{2}\pi_{\alpha}\pi^{\alpha}g_{\nu\mu}+\frac{1}{2}\pi_{\mu}\pi_{\nu}-\pi_{\nu}\pi_{\mu}\\
        =\overset{\circ}{R}_{\mu \nu}&-\frac{1}{2}g_{\mu\nu}\overset{\circ}{\nabla}_\alpha\pi^{\alpha}+\overset{\circ}{\nabla}_\mu\pi_{\nu}+\overset{\circ}{\nabla}_\nu \pi_{\mu}\\
        &+\pi_{\rho}\pi^{\rho}g_{\mu\nu}+\frac{1}{2}\pi_{\mu}\pi_{\nu}.
\end{aligned}
\end{equation*}
\twocolumngrid
As a summary, for the Ricci tensor of the dual connection we obtain
\begin{equation}
\begin{aligned}
    R^{*}_{\mu \nu}=\overset{\circ}{R}_{\mu \nu} &- \frac{1}{2} g_{\mu \nu} \overset{\circ}{\nabla}_{\alpha} \pi^{\alpha}+\overset{\circ}{\nabla}_{\mu} \pi_{\nu} + \overset{\circ}{\nabla}_{\nu} \pi_{\mu}\\
    &+ \pi_{\rho} \pi^{\rho} g_{\mu \nu} + \frac{1}{2} \pi_{\mu} \pi_{\nu}.
\end{aligned}
\end{equation}

It immediately follows that the dual Ricci scalar is given by
\begin{equation}
R^{*}=\overset{\circ}{R}+ \frac{9}{2} \pi^{\rho} \pi_{\rho}.
\end{equation}

\section{Derivation of the Friedmann equations} \label{appendixC}
First, recall that the non-zero components of the Ricci tensor of the Levi-Civita connection are given by
\begin{equation}
    \overset{\circ}{R}_{00}=-3\frac{\ddot a}{a}, \; \; \overset{\circ}{R}_{11}=\overset{\circ}{R}_{22}=\overset{\circ}{R}_{33}=a \ddot a +2 \dot a^2.
\end{equation}

Similarly, the Ricci scalar takes the well known form
\begin{equation}
    \overset{\circ}{R}=6 \left(\frac{ \ddot a}{a}+\frac{\dot a ^2}{a^2} \right).
\end{equation}

The non-vanishing Christoffel symbols are given by
\begin{equation}
    \tensor{\gamma}{^0_i_j}=a \dot{a} \delta_{ij}, \; \; i,j=1,2,3;
\end{equation}
\begin{equation}
    \tensor{\gamma}{^i_0_j}=\frac{\dot a}{a} \delta^{i}_{j}, \; \; i,j=1,2,3.
\end{equation}

As a first step, we will write down the Einstein equation in $00$ components
\begin{equation}
\begin{aligned}
    \overset{\circ}{R}_{00}& - \frac{1}{2} g_{00} \overset{\circ}{R} + \frac{3}{8} g_{00} \overset{\circ}{\nabla}_\alpha \pi^{\alpha} + \frac{3}{8} \overset{\circ}{\nabla}_0 \pi_0 + \frac{3}{8} \overset{\circ}{\nabla}_{0} \pi_0\\
    &+ \frac{5}{16} g_{00} \pi_\alpha \pi^\alpha - \frac{1}{8} \pi_0 \pi_0= 8 \pi T_{00}.
    \end{aligned}
\end{equation}

From these, we obtain
\begin{equation}
    3 \frac{\dot a^2}{a^2} - \frac{3}{8} \left( \dot \psi + 3 \frac{ \dot a}{a} \psi \right) -\frac{3}{8} \dot \psi - \frac{3}{8} \dot \psi +\frac{5}{16} \psi^2 - \frac{1}{8} \psi^2=8 \pi \rho.
\end{equation}

We thus have
\begin{equation}
    3\frac{\dot a^2}{a^2} -\frac{9}{8} \dot \psi - \frac{9}{8} \frac{\dot a}{a} \psi+\frac{3}{16} \psi^2=8 \pi \rho.
\end{equation}

For the second Friedmann equation, we write down the $ii$ components
\begin{equation}
\begin{aligned}
    \overset{\circ}{R}_{ii} &- \frac{1}{2} g_{ii} \overset{\circ}{R} + \frac{3}{8} g_{ii} \overset{\circ}{\nabla}_{\alpha} \pi^{\alpha} + \frac{3}{8} \overset{\circ}{\nabla}_{i} \pi_i + \frac{3}{8} \overset{\circ}{\nabla}_i \pi_i \\
    &+\frac{5}{16} g_{ii} \pi_{\alpha} \pi^{\alpha} - \frac{1}{8} \pi_i \pi_i=8 \pi T_{ii}.
\end{aligned}
\end{equation}

This immediately yields
\begin{equation}
\begin{aligned}
    -2 a \ddot a - \dot a^2 &+ \frac{3}{8}a^2 \left( \dot \psi + 3 \frac{\dot a}{a} \psi \right)\\
    &+ \frac{3}{8} a \dot a \psi + \frac{3}{8} a \dot a \psi - \frac{5}{16} a^2 \psi^2=8 \pi p a^2.
\end{aligned}
\end{equation}

Simplifying, we get
\begin{equation}
    -2 a \ddot a - \dot a^2+\frac{3}{8} a^2 \dot \psi +\frac{15}{8} a \dot a \psi - \frac{5}{16} a^2 \psi^2 = 8 \pi p a^2.
\end{equation}

Dividing by $a^2$ leads to
\begin{equation}
    -2\frac{ \ddot a}{a} - \frac{\dot a^2}{a^2} + \frac{3}{8} \dot \psi + \frac{15}{8} \frac{\dot a}{a} \psi - \frac{5}{16} \psi^2=8\pi p.
\end{equation}

Introducing the Hubble function, we obtain
\begin{equation}
    -2 \dot H -3H^2 + \frac{3}{8} \dot \psi + \frac{15}{8} H \psi - \frac{5}{16} \psi^2=8 \pi p.
\end{equation}

The final form is thus
\begin{equation}
    2 \dot H + 3 H^2= -8 \pi p + \frac{3}{8} \dot \psi + \frac{15}{8} H \psi - \frac{5}{16} \psi^2.
\end{equation}
\newpage
\bibliographystyle{unsrt}

\end{document}